\documentclass[12pt]{article}

\usepackage{latexsym,graphics,epsfig,amsmath,amssymb,tabularx,booktabs, color}
\usepackage[ngerman,english]{babel}

\usepackage{natbib}
\usepackage{enumerate}
\usepackage{amsthm}
\usepackage{kantlipsum}
\allowdisplaybreaks
\usepackage{chngcntr}
\usepackage[shortlabels]{enumitem}
\usepackage[utf8]{inputenc}

\usepackage[mathscr]{eucal}

\usepackage[T1]{fontenc}

\theoremstyle{definition}
\newtheorem{theorem}{Theorem}[section]
\newtheorem{lemma}[theorem]{Lemma}
\newtheorem{corollary}[theorem]{Corollary}

\newcommand{\true}[0]{\oset{*}{\boldsymbol \lambda}}
\newcommand{\ti}[0]{\oset{*}{\lambda_i}}

\newcommand{\ptrue}[0]{\oset{*}{\boldsymbol p}}
\newcommand{\pti}[0]{\oset{*}{p_i}}

\newcommand{\ptone}[0]{\oset{*}{p_1}}
\newcommand{\ptI}[0]{\oset{*}{p_I}}

\DeclareFontFamily{OT1}{pzc}{}
\DeclareFontShape{OT1}{pzc}{m}{it}{<-> s * [1.200] pzcmi7t}{}
\DeclareMathAlphabet{\mathpzc}{OT1}{pzc}{m}{it}

\makeatletter
\newcommand{\oset}[2]{%
  {\mathop{#2}\limits^{\vbox to -.1\ex@{\kern-\tw@\ex@
   \hbox{\scriptsize #1}\vss}}}}

\newenvironment{keywords}{%
\begin{changemargin}{1cm}{1cm}
\noindent{\bf Keywords:}
}
{\end{changemargin} }

\newenvironment{changemargin}[2]{%
\begin{list}{}{%
\setlength{\topsep}{0pt}%
\setlength{\leftmargin}{#1}%
\setlength{\rightmargin}{#2}%
\setlength{\listparindent}{\parindent}%
\setlength{\itemindent}{\parindent}%
\setlength{\parsep}{\parskip}%
}%
\item[]}{\end{list}}

\begin{document}

\title{Testing the fit of relational models}

\author{Anna Klimova \\
{\small{National Center for Tumor Diseases (NCT), Partner Site Dresden, and}}\\
{\small{Institute for  Medical Informatics and Biometry,}}\\ 
{\small{Technical University, Dresden, Germany} }\\
{\small \texttt{anna.klimova@nct-dresden.de}}\\
{}\\
\and 
Tam\'{a}s Rudas \\
{\small{Department of Statistics, E\"{o}tv\"{o}s Lor\'{a}nd University, Budapest, Hungary}}\\
{\small \texttt{trudas@elte.hu}}\\
}

 \date{}
   
\maketitle

\begin{abstract}
\noindent Relational models generalize log-linear models to arbitrary discrete sample spaces by specifying effects associated with any subsets of  their cells. A relational model may include an overall effect, pertaining to every cell after a reparameterization, and in this case, the properties of the maximum likelihood estimates (MLEs) are analogous to those computed under traditional log-linear models, and the goodness-of-fit tests are also the same. If an overall effect is not present in any reparameterization, the properties of the MLEs are considerably different, and the Poisson and multinomial MLEs are not equivalent. In the Poisson case, if the overall effect is not present, the observed total is not always preserved by the MLE,  and thus, the likelihood ratio statistic is not identical with twice the Kullback-Leibler divergence. However, as demonstrated, its general form may be obtained from the Bregman divergence. The asymptotic equivalence of the Pearson chi-squared and likelihood ratio statistics holds, but the generality considered here requires extended proofs. 
\end{abstract}


\section{Introduction}\label{intro}

Relational models \citep*{KRD11} are a variant of the log-linear model, which may be applied to discrete sample spaces with or without a Cartesian product structure. These sample spaces still may be  multivariate, being  subsets of Cartesian products of the ranges of categorical variables. The general form of a relational model is 
$$
\log \boldsymbol{\delta} = \mathbf{A}' \boldsymbol{\theta},
$$
where $\boldsymbol{\delta}$ are parameters associated with the cells $i$ of the sample space $\mathcal{I}$. They can be probabilities $(\boldsymbol{p})$, when the observational procedure works with a pre-specified sample size, or may be intensities $(\boldsymbol{\lambda})$, when such a pre-specified sample size does not exist and the sample has a Poisson distribution. In the case of traditional log-linear models, such a distinction is not needed, but it is necessary in the generality considered here. 

The design matrix $\mathbf{A}$  is a $0-1$ full row rank matrix, with at least one $1$ in every column, and $\boldsymbol{\theta}$ are the model parameters. In traditional log-linear models, the rows of $\mathbf{A}$ are indicators of cylinder sets of marginals of $\mathcal{I}$, and the parameters of the model are associated with these cylinder sets. In the present case, $\mathbf{A}$ may have any structure and the parameters may be associated with arbitrary subsets of $\mathcal{I}$. Each model parameter $\boldsymbol{\theta}$ quantifies a linear effect on the logarithms of the parameters  $\boldsymbol{\delta}(i)$ of the cells $i \in \mathcal{I}$, which belong to its subset. These subsets may also be seen as defining the model, because the rows of $\mathbf{A}$ are their indicators.

Earlier work on relational models includes  \cite{Anna} and \cite{KRipf1}, where a generalized iterative scaling procedure for computing maximum likelihood estimates under relational models was proposed and its convergence was proved. The algorithm was implemented in \cite{gIPFpackage}. Existence of the MLEs, including the cases when there are some zeros in the observed data, was studied in \cite{KRextended}. A unique MLE to such data was shown to always exist in  the closure of the original model with respect to the Bregman information divergence. The same variant of iterative scaling may be used to compute the MLE whether it is in the original model or in its closure. The relational model class was applied by \cite{KRbm} to study trends in social mobility. {\cite{KRipf1} contains further examples when relational models in the generality considered here need to be applied, including Aitchison-Silvey independence.}

An important characteristic  of relational models is that the existence of a common effect present in all cells is not assumed. When the row space of the design matrix contains a vector of $1$'s, the model is said to contain an overall effect and a straightforward reparameterization may be applied to explicitly show it.  When an overall effect is not present, relational models assuming multinomial or Poisson sampling turn out to be very different from each other. Maximum likelihood estimates based on a data set, when the model is applied to probabilities or to intensities, are not equivalent. For example, the MLEs under multinomial sampling do not preserve the observed sums of the subsets defining the model, are, instead,  proportional to them, see Example 4.2 in \cite{KRD11}. MLEs for Poisson samples, on the other hand, do not preserve the observed total. This fact makes some of the standard testing techniques invalid. 

A non-technical overview of the reasons why relational models need to be considered to solve many real statistical problems is given in Section \ref{Relevance}. It is illustrated, why the incomplete Cartesian product structure of the sample space and the lack of an overall effect are important in modeling and how these characteristics are related.

Section \ref{Estimation} of this paper summarizes the main properties of MLEs for relational models and considers the problem of testing model fit. It is shown that the likelihood ratio statistic has a form different from the one under standard log-linear models. This general form is derived from the Bregman divergence \citep{Bregman}, which was used in \cite{KRipf1} to prove the convergence of the general variant of Iterative Proportional Fitting which yields maximum likelihood estimates under relational models. To refer to this, the likelihood ratio statistic will be called the Bregman statistic in this context.

Section \ref{AsympNormSection} gives a proof of the asymptotic normality of the maximum likelihood estimates in the generality considered here. This is used in the next section to prove results about the asymptotic behaviour of test statistics. Of course, as is usual, the asymptotic normality is obtained after appropriate scaling. The scaling factor in the case of probabilities, is the square root of the sample size, and in the case of intensities, inversely proportional to the square root of the true parameter. 

Section \ref{TestStatSection} proves that when the relational model holds, the asymptotic distributions of the Pearson and of  the Bregman statistics are chi-squared, and the number of degrees of freedom is derived, too. Finally, Section \ref{SimulationsSection} illustrates the asymptotics with Monte Carlo simulations and sheds some light on the issue of the speed of convergence to the asymptotic distribution. Proofs not included in the paper are provided in the Appendix. 

In the remainder of this section, the relationship between the existing literature and  the results presented in the paper is discussed.
When the overall effect is present, the relational model is defined by fixed values of homogeneous odds ratios \citep{KRD11}, although these odds ratios may be more complex, than the ones considered in the literature so far. This implies that, if a vector of intensities is in the model, then a multiple of this vector  is in the model too, leading to the equivalence of the likelihood analysis for Poisson and multinomial samples. A setup implying the presence of the overall effect  was considered by \cite{LangPoissMult}. Among other results, he proved asymptotic normality of the MLEs under various sampling schemes and derived their covariance matrices. The asymptotics in the multinomial case were computed with respect to the sample size approaching infinity, and in the Poisson case under the assumption that the cell parameters go to infinity at the same rate. Lang did not discuss specifically the asymptotic behaviour of the test statistics, but briefly mentioned that the known equivalence between goodness-of-fit statistics holds.

When the overall effect is not present, the relational model is defined by the fixed values of at least one non-homogeneous and some homogeneous odds ratios. Therefore, the scale invariance described above does not hold, and  when a probability vector is in the model, the corresponding frequency vector is not in the model. This implies that the asymptotic setup based on the equal rates of convergence, considered by \cite{LangPoissMult}, does not apply. Instead, the asymptotics considered here will  assume that even the smallest of the intensities converges to infinity, but the vectors of intensities remain in the model, for the Poisson case. For the multinomial case, a fixed probability vector, and the sample size converging to infinity is the asymptotic setup used.

The asymptotic distributions of the MLEs under relational models with the overall effect can be determined from the results of Lang. However, the present paper also derives the asymptotic distributions  without relying on the presence of the overall effect. As  shown in the Appendix, the asymptotic covariance matrix obtained in this paper simplifies for the case with the overall effect, and then coincides with the result of Lang, if the latter is applied to a multinomial or Poisson distribution.

Many of the proofs given are modeled upon ideas presented in  \cite{AitchSilvey60}. However, the discussion in that paper is quite heuristic, while we give rigorous proofs. Further, as opposed to this paper, they do not consider the Bregman statistic and do not derive the number of degrees of freedom associated with the asymptotic distributions of the test statistics.

\section{On the scope and relevance of this work}\label{Relevance}

In this section, the main motivation for and possible uses of the results of this paper are discussed in a non-technical way.

Traditionally in multivariate statistics, the sample space is the Cartesian product of the ranges of the variables of interest and the relevant probability space is identified with this sample space. When the variables are continuous, it is obvious that not all possible combinations of their values may be observed, but in the discrete case, cells with no observations call for caution. There is a somewhat informal theory of handling empty cells, which distinguishes between structural and sampling (or observed or random) zero frequencies \citep[cf.][]{Mohri}. The standard procedure in applied work  is to reduce the number if degrees of freedom in chi-squared based testing by the number of structural zeros (see \cite{NemethBSM} for an implementation of this procedure). If possible, however, one should avoid models which associate positive probabilities with cells which do not exist (see \cite{KRD11} for comments related to quasi models and relational models). Another related stream of work applies zero-inflated models to deal with the issue of structural or sampling zeros \citep*[cf.][]{HeZeros}.

Similar setups appear in such areas of real statistical applications, as text processing \citep*[cf.][]{Mccallum2000, Mohri},  data mining of transactional data \citep*[cf.][]{WuChenHanMeasures}, computer tomography \citep[cf.][]{OSullivan}, and market basket analysis \citep{Giudici}, where the research focuses on studying associations between some kind of binary features, or attributes. A number of features  can be present or absent in an object in various combinations, but the no-feature-present combination is not possible. A hypothetical example is given by \cite{AitchSilvey60}, where a model of independence between three features was proposed for this context. Another important area where such designs may occur is the study of dose limiting toxicity from a combination of drugs. For example, in clinical trials in oncology  \citep*[cf.][]{LeTourneauPhase1}, all patients get some treatment, so there is no observation for untreated patients, but also not all combinations of treatments are applied. In general, this is also the case in all observational studies using incomplete factorial designs. In the next Section, a variant of independence of Aitchison and Silvey  is applied to such a study.
 
 Another relevant example is the calf pneumonia problem described by \cite{Agresti2013}. At the first stage of the experiment, a group of calves was exposed to a pneumonia infection. At the second stage, In order to see whether the primary infection had an immunizing effect, the calves who got the primary infection were exposed to the infection again. So, a suitable model to test is the independence of the two infections. Because having the secondary infection without having the primary one is logically impossible, it is not the data collection design, rather the reality has an incomplete Cartesian product structure. These considerations suggest to extend the traditional structural or observed zero theory to a much larger number of possibilities, including empty cell logically, or empty cell not logically but in a particular population where the data come from, or empty cell because of sampling, or cell not observed during the data collection process which does not exist logically, or does not exist in the population, or does exist in the population, or its existence is unknown in the population. Many of these require modeling probability distributions on incomplete Cartesian products.

\section{Estimation and testing in relational models}\label{Estimation}

As a simple example where relational models may be relevant, consider the bait study for swimming crabs by \cite*{Kawamura1995}. In this study, three types of baits were used: fish alone, sugarcane alone, fish and sugarcane together. The numbers of crabs caught in each trap were $11$, $2$ and $36$, respectively. The study design is shown in Table 1.
 

\begin{table}
\begin{minipage}[t]{0.45\textwidth}
\centering
\begin{tabular}[t]{ccc}
 &
\multicolumn{2}{c}{ \mbox{Sugarcane }}\\
\cline{2-3} \\ [-6pt]
\mbox{ Fish}& \mbox{Yes}& \mbox{No }\\
\hline \\ [-6pt]
\mbox{ Yes }	& $\lambda_{11}$	& $\lambda_{10}$\\
\mbox{ No } 	& $\lambda_{01}$  &   -  \\[6pt]
\end{tabular}
\caption{Study design for swimming crabs}
\label{Tab1}
\end{minipage} \qquad
\begin{minipage}[t]{0.45\textwidth}
\centering
\begin{tabular}[t]{ccc}
 &
\multicolumn{2}{c}{ \mbox{ Sugarcane }}\\ 
\cline{2-3}\\ [-6pt]
\mbox{ Fish }& \mbox{Yes}& \mbox{No }\\
\hline \\ [-6pt]
\mbox{ Yes }	& $\lambda_{10} \lambda_{01}$	& $\lambda_{10}$\\
\mbox{ No } 	& $\lambda_{01}$  &   -  \\ [6pt]
\end{tabular}
\caption{A variant of independence for swimming crabs}
\label{Tab2}
\end{minipage}

\end{table}

As the number of crabs was not decided in advance, three Poisson random variables may be used to model the number of crabs caught in the traps. The sample space is a proper subset of the $2 \times 2$ contingency table formed by the ranges (yes or no) of the variables indicating the presence or lack of fish and sugarcane as bait. Although a routine question in the statistical analysis of bivariate data, the hypothesis of independent effects of fish and sugarcane cannot be tested here. Traditional independence is only meaningful if one  assumes that crabs would also enter a trap with no bait, and then independence would be proportionality in the resulting complete $2 \times 2$ table. If this assumption is not made, that is, one believes  that no crab would enter an empty trap, the no-fish and no-sugarcane cell becomes nonexistent, and independence, which requires a complete $2 \times 2$ table, makes no sense in this case.
If the assumption is made, the empty cell becomes an unobserved cell, and may always be filled with a number, such that independence holds, irrespective of the observed data. Thus, independence cannot be tested in this case. 

A possible relational model which is a variant of independence (in fact, a variant of the Aitchison-Silvey independence, see \cite{AitchSilvey60}) assumes that
$$
\log \left(
\begin{array}{c} 
\lambda_{10} \\
\lambda_{01} \\
\lambda_{11} \\
\end{array}
\right)
=
\left(
\begin{array}{c c}
1 & 0 \\
0 & 1 \\
1 & 1 \\
\end{array}
\right)
\left(
\begin{array} {c}
\theta_1 \\
\theta_2 \\
\end{array}
\right) .
$$

This model may be written in the form shown in Table \ref{Tab2}, and is a relational model which is a kind of independence and may be defined by setting the  
non-homogeneous odds ratio  $\,{\lambda_{10} \lambda_{01}}/{\lambda_{11}}\,\, $ to $1$ \citep{KRD11}.
If one included an overall effect present in all $3$ cells,  the model would become non-restrictive (sometimes called saturated).

Methods of obtaining MLEs for relational model were discussed in \cite{KRipf1} and an R-function \citep{gIPFpackage} is also available.  The MLEs for the crab data are $11.94$, $2.94$, and $35.06$. The observed subset sums $11+36=47$ and $2+36=38$ are preserved by the MLE: $11.94+35.06=47$ and $2.94+35.06=38$. However, somewhat surprisingly, the observed total is not preserved: $11+2+36 \neq 11.94+2.94+35.06$.

A summary of the properties of the MLEs under relational models is given in Table \ref{ModelTypes}.
When the overall effect is present, the subset sums and total are preserved by the MLEs for both probabilities and intensities. When the overall effect is not present, for probabilities, the total is preserved by the MLEs, but the subset sums are not equal but rather proportional to the observed subset sums, and for intensities, the subset sums are preserved, but the total is not preserved.

\begin{table}[ht]
\centering
\vspace{2mm}
\setlength{\extrarowheight}{5pt}
{
\begin{tabular}{m{32mm}ccc}
&
\multicolumn{1}{c}{{Models with the overall effect}} &
\multicolumn{2}{c}{ {Models without the overall effect}}\\
\cmidrule(lr){2-2} \cmidrule(lr){3-4}
 & \multicolumn{1}{c}{Probabilities $\&$ Intensities} & \multicolumn{1}{c}{ Probabilities} & \multicolumn{1}{c}{ Intensities}\\
\hline
{ Subset sums of the MLE vs observed subset sums }&  \multicolumn{1}{c}{Equal} 	& \multicolumn{1}{c}{{Proportional}} 	& \multicolumn{1}{c}{Equal}\\
\hline
{ Total of the MLE vs observed total }&  \multicolumn{1}{c}{Equal} 	& \multicolumn{1}{c}{{Equal}} 	& \multicolumn{1}{c}{Not equal}\\ 
\end{tabular}}
\vspace{3mm}
\caption{Relational models and properties of the MLE. Adapted from  \protect\cite{KRipf1}.}
\label{ModelTypes}
\end{table}

To subject the intuitively good fit seen in the crabs data to a formal test, one may compute the Pearson chi-squared statistic, and the value of $0.4$ is obtained. But it is entirely unclear, whether the asymptotic chi-squared reference distribution may be used, given that the total was not preserved. And even if the reference distribution may be used, the appropriate  number of degrees of freedom is not known. As far as the likelihood ratio statistic is concerned, one obtains:
\begin{align*}
&2 (\lambda_{10} \cdot \log ({\lambda_{10}}/{\hat{\lambda}_{10}})  + \lambda_{01} \cdot \log ({\lambda_{01}}/{\hat{\lambda}_{01}})+ \lambda_{11} \cdot \log ({\lambda_{11}}/{\hat{\lambda}_{11}}) \\ &- [(\lambda_{10} + \lambda_{01} + \lambda_{11}) 
- (\lambda_{10} + \hat{\lambda}_{01} + \hat{\lambda}_{11} )] ) \\
& = 2 (11\cdot\log(11/11.94) + 2\cdot \log(2/2.94) + 36 \cdot \log(36/35.06)- (49 - 49.94)) = 0.44. 
\end{align*}
Note that because the observed and estimated totals are not necessarily equal, the $[(\lambda_{10} + \lambda_{01} + \lambda_{11}) 
- (\lambda_{10} + \hat{\lambda}_{01} + \hat{\lambda}_{11} )]$ term does not cancel.

One sees that the likelihood ratio statistic is twice the Bregman \citep{Bregman} divergence between the observed and the estimated probabilities. The Bregman divergence between two positive vectors is defined as: 
$$
\mathcal{B}(\boldsymbol t, \boldsymbol u) = \sum_i t_i \log \frac{t_i}{u_i}  - \sum_i (t_i-u_i),
$$
and is a generalization of the Kullback - Leibler divergence. The Bregman statistic, when the observed data are $\boldsymbol y$ and the MLEs are $\hat{\boldsymbol y}$, is
$$
B(\boldsymbol y, \hat{\boldsymbol y}) = 2\mathcal{B}(\boldsymbol y, \hat{\boldsymbol y}) = 2 \left(\sum_i y_i \log \frac{y_i}{\hat{y}_i}  - \sum_i ( y_i - \hat{y}_i) \right).
$$

The rest of the paper proves that the Pearson and the Bregman statistics have asymptotic chi-squared distributions. As the initial step, the asymptotic normality of the MLE under relational models is obtained.

\section{Asymptotic normality of the MLE}\label{AsympNormSection}

Asymptotic properties of the MLE under a model for intensities  are studied first. Let $\boldsymbol \lambda = (\lambda_1, \dots, \lambda_I)'\in \mathbb{R}^I_{>0}$, and $\boldsymbol Y = (Y_1, \dots, Y_I)$ denote a random vector that has a multivariate Poisson distribution on the sample space $\mathcal{I}$: $\,\,\boldsymbol Y \sim Pois(\boldsymbol \lambda)$. Namely, for each  $i = 1, \dots, I$, the random variable $Y_i$ has a Poisson distribution, $Y_i \sim Pois(\lambda_{i})$, and $Y_1, \dots, Y_I$ are jointly independent. The asymptotic distribution of the MLE in a model for intensities is obtained under the assumption that the minimum norm, $\|\boldsymbol \lambda\| = \underset{i = 1, \dots, I}{\min} \lambda_i$, goes to infinity. 

Let $\{\boldsymbol \lambda^{r}\}_{r=1}^{\infty} \subset \mathbb{R}^I_{>0}$, such that $\|\boldsymbol \lambda^{r}\| \to \infty$ as $r \to \infty$. For each $r$, let $\boldsymbol Y^{r} \sim Pois(\boldsymbol \lambda^{r})$.  In the sequel, for simplicity of exposition, the index $r$ is omitted, and a shorthand notation is used: $\boldsymbol Y \equiv \boldsymbol Y^{r}$, $\true \equiv \boldsymbol \lambda^{r}$, and $\hat{\boldsymbol \lambda} \equiv \hat{\boldsymbol \lambda}{}^{r}$. The convergence of the norm $\|\boldsymbol \lambda^{r}\|$ to infinity when $r \to \infty$ is abbreviated as $\|\true\| \to \infty$.   Throughout the paper, the $I \times I$ identity matrix is denoted by $\mathbf I$, and $\Delta[\boldsymbol \lambda]$ stands for the diagonal matrix with $\boldsymbol \lambda$ as its main diagonal.

\begin{lemma}\label{DataNormala}
Let $\boldsymbol Y \sim Pois(\true)$. Then, as
 $\|\true\| \to \infty$, 
\begin{equation}\label{PoissonNormal}
\mathcal{L}\left\{\Delta[{\true}{}^{-1/2}]({\boldsymbol Y} - \true)\right\} \quad {\to} \quad \mathcal{N}(\boldsymbol 0, \mathbf{I}).
\end{equation}
\end{lemma}


\begin{proof}
For each $i = 1, \dots, I$, 
\begin{equation}
\mathcal{L} \left\{\ti{}^{-1/2}(Y_i - \ti)\right\} \quad {\to} \quad \mathcal{N}(0, 1),
\end{equation}
as $\ti \to \infty$ \citep*[cf.][p. 489]{BFH}. 
The required result follows from the joint independence of $Y_1, \dots, Y_I$. 
\end{proof}

\vspace{1mm}

\begin{corollary} \label{CorollaryIntensNormal} Let $\boldsymbol Y \sim Pois(\true)$ and $\boldsymbol y$ be a realization of $\boldsymbol Y$. Then,
as $\|\true\| \to \infty$,
\begin{align}\label{Op}
&\Delta[{\true}{}^{-1/2}]({\boldsymbol Y} - \true) = O_p(\boldsymbol 1), \\
&\Delta[{\true}{}^{-1/2}]({\boldsymbol y} - \true) = O(\boldsymbol 1).\nonumber
\end{align}
\end{corollary}

\vspace{2mm}

The asymptotic normality of the MLE under a relational model for intensities will be obtained next. The relevant scaling factor in the case of intensities is 1 over the square root of the true parameter (and in the case of probabilities is the square root of the sample size), but this scaling will not always be made explicit when asymptotic normality is mentioned in the text.
The proof of this result will use two lemmas, whose proofs are given in the Appendix. The first lemma describes the asymptotic proximity of the MLE to the true parameter value:

\begin{lemma}\label{newlemma} 
Let $RM_{\boldsymbol \lambda}(\mathbf{A})$ be a relational model for intensities, $\true \in RM_{\boldsymbol \lambda}(\mathbf{A})$ be the true parameter, and $\boldsymbol Y \sim Pois(\true)$ be the observations.  Assume that the MLE $\hat{\boldsymbol \lambda}$ of $\true$ under the model  exists. Then 
\begin{equation}\label{majorY}
\Delta[\true{}^{-1/2}] (\hat{\boldsymbol \lambda} - \true) = O_p(\boldsymbol 1), \quad \mbox{ as } \|\true\| \to \infty.
\end{equation}
\end{lemma}

\noindent The second lemma describes the asymptotic behaviour of certain diagonal matrices, built from the vectors located close to the true value: 

\begin{lemma}\label{smallLemma}
Let $\boldsymbol \lambda > \boldsymbol 0$ be a random variable such that $\,\Delta[\true{}^{-1/2}](\boldsymbol \lambda - \true) =O_p(\boldsymbol 1),$ 
as $\|\true\| \to \infty$, and consider $\boldsymbol \varsigma > \boldsymbol 0$ that satisfy
$\,\|\boldsymbol \lambda - \boldsymbol  \varsigma\| < \|\boldsymbol \lambda - \true\|$. Then, for $m \in \mathbb{Z}_{>0}$,  
\begin{align*}
\Delta[\true{}^m]\cdot\Delta[\boldsymbol \varsigma^{-m}] &= \mathbf{I}+ O_p(\|\true\|^{-1/2}) \qquad \mbox{ as } \|\true\| \to \infty.
\end{align*}
\end{lemma}

\vspace{4mm}

\noindent If the observed data are strictly positive, the MLE always exists, but, if the data contain some zeros, the MLE may or may not exist, depending on the pattern of zeros \citep{KRextended}. In order to deal with this situation, an augmented MLE, $\tilde{\boldsymbol \lambda}$, and augmented Lagrange multipliers, $\tilde{\boldsymbol \alpha}$, are introduced. A similar approach was used in \cite{WittingNoelle}, among others. Let $\mathscr{E}$ denote the set of values of $\boldsymbol Y$, for which the MLE exist, and $\bar{\mathscr{E}}$ denote its complement in $\mathbb{Z}_{\geq 0}$. The augmented MLE and Lagrange multipliers are defined to be equal, respectively, to $\hat{\boldsymbol \lambda}$ and $\hat{\boldsymbol \alpha}$, if the data are in $\mathscr{E}$, and to a positive vector and a zero vector otherwise, 
namely:
\begin{equation}\label{pseudoMLE_intens}
\tilde{\boldsymbol \lambda} = \tilde{\boldsymbol \lambda}  (\boldsymbol Y)  =\left\{\begin{array}{ll}
 \hat{\boldsymbol \lambda}(\boldsymbol Y),   & \mbox{ if } \boldsymbol Y \in \mathscr{E}, 
 \\
 (\boldsymbol 1'\boldsymbol Y/I)\boldsymbol 1,  & \mbox{ if } \boldsymbol Y \in \bar{\mathscr{E}}, 
 \end{array} \right.   \quad \tilde{\boldsymbol \alpha} = \tilde{\boldsymbol \alpha}  (\boldsymbol Y)  =\left\{\begin{array}{ll}
 \hat{\boldsymbol \alpha}(\boldsymbol Y), & \mbox{ if } \boldsymbol Y \in \mathscr{E}, 
 \\
 \boldsymbol 0,  & \mbox{ if } \boldsymbol Y \in \bar{\mathscr{E}}.
 \end{array} \right.
\end{equation}
The main result about the asymptotic normality is presented now:

\begin{theorem}\label{MLENormala} 
Let $RM_{\boldsymbol \lambda}(\mathbf{A})$ be a relational model for intensities, $\mathbf{D}$ be a kernel basis matrix, $\true \in RM_{\boldsymbol \lambda}(\mathbf{A}) $, and $\boldsymbol Y \sim Pois(\true)$. Let $\hat{\boldsymbol \lambda}$ be the MLE of $\true$ under $RM_{\boldsymbol \lambda}(\mathbf{A})$ and $\tilde{\boldsymbol \lambda}$ be the augmented MLE.  
Then, as $ \|\true\|\to \infty$,
\begin{enumerate}[(i)]
\item 
$\mathcal{L} \left(\Delta[\true{}^{-1/2}]({\hat{\boldsymbol \lambda}} - \true) \,\, {\boldsymbol{\mid}} \, \boldsymbol Y \in \mathscr{E} \right) \quad {\to} \quad \mathcal{N
}(\boldsymbol 0, \mathbf{I} -\mathbf{D}'(\mathbf{D}\mathbf{D}')^{-1}\mathbf{D}),
$ \\

\item $\mathcal{L} \left(\Delta[\true{}^{-1/2}]({\tilde{\boldsymbol \lambda}} - \true) \right) \quad {\to} \quad \mathcal{N
}(\boldsymbol 0, \mathbf{I} -\mathbf{D}'(\mathbf{D}\mathbf{D}')^{-1}\mathbf{D}).
$ 
\end{enumerate}
\end{theorem}

\noindent The conditional asymptotic normality of the MLE under relational models for intensities with the overall effect can be derived from the results of \cite{LangPoissMult}, Section 5, but his approach cannot be extended to the no-overall-effect situation. The proof presented below applies to models with and without the overall effect.


\begin{proof} 
Let $l(\boldsymbol \lambda, \boldsymbol Y)$ denote the log-likelihood function of a Poisson distribution: 
$$l(\boldsymbol \lambda, \boldsymbol Y) = \boldsymbol Y'  \log \boldsymbol \lambda  - \boldsymbol 1'\boldsymbol \lambda.$$
Assume that $\boldsymbol Y \in \mathscr{E}$, and thus, the MLE $\hat{\boldsymbol \lambda}$ under the relational model exists and is the unique solution to the following maximization problem  \citep{KRextended}: 
$$ \underset{\boldsymbol \lambda \in \mathcal{D}}{\mbox{max }} l(\boldsymbol \lambda, \boldsymbol Y),$$
where $\mathcal{D} = \{ \boldsymbol \lambda \in \mathbb{R}_{>0}^{I}: \,\,\mathbf{D} \log \boldsymbol \lambda = \boldsymbol 0\}$. In order to find the asymptotic distribution of $\hat{\boldsymbol \lambda}$, an approach analogous to the one described in \cite{AitchSilvey58} is applied. It utilizes the fact that the MLE is the unique value of $\boldsymbol \lambda$ that maximizes the Lagrange function,
\begin{equation}
L(\boldsymbol \lambda, \boldsymbol \alpha, \boldsymbol Y) =\boldsymbol Y'  \log \boldsymbol \lambda  - \boldsymbol 1'\boldsymbol \lambda+ \boldsymbol \alpha' \mathbf{D} \log\boldsymbol \lambda,
\end{equation}
over the unrestricted parameter space $\boldsymbol \lambda \in \mathbb{R}_{>0}^{I}$ \citep{KRextended}. 
Here $\boldsymbol \alpha = (\alpha_1, \dots, \alpha_K)'$ denotes the column-vector of Lagrange multipliers. Consequently,  $\hat{\boldsymbol \lambda}$ and the corresponding values of Lagrange multipliers $\hat{\boldsymbol \alpha}$ satisfy the set of equations obtained by differentiating the Lagrange function with respect to the components of $\boldsymbol \lambda$ and of $\boldsymbol \alpha$:
\begin{align}\label{likeEqf}
{\partial L}/{\partial \boldsymbol \lambda} &= \Delta[\boldsymbol \lambda^{-1}] \boldsymbol Y  - \boldsymbol 1 + \Delta[\boldsymbol \lambda^{-1}] \mathbf{D}'{\boldsymbol \alpha} = \boldsymbol 0,\\
{\partial L}/{\partial \boldsymbol \alpha} &= \mathbf{D} \log\boldsymbol \lambda = \boldsymbol 0. \nonumber
\end{align}
The argument below develops an approximation, as $\|\true\| \to\infty$, to the non-stochastic version of (\ref{likeEqf}), given in (\ref{systemPnewMatrix_estimate_inverse}). The solution given in (\ref{AsympNewNS}) approximates the solution of the non-stochastic version of (\ref{likeEqf}), for every value $\boldsymbol y$. But (\ref{AsympNewNS}) is also true in the stochastic sense. Therefore, the same approximation, given in (\ref{Asymp}), applies to the solution of (\ref{likeEqf}). Finally, the delta method is used to derive the asymptotic distribution of the solution, given in (\ref{CondIntens}).

Let $\boldsymbol Y = \boldsymbol y$ and consider the non-stochastic version of (\ref{likeEqf}). By the Taylor theorem, for every $\boldsymbol \lambda$, there exist such $\boldsymbol \phi = \boldsymbol \phi(\true, \boldsymbol \lambda, \boldsymbol y)$, $\,\boldsymbol \psi = \boldsymbol \psi(\true, \boldsymbol \lambda, \boldsymbol y)$, $\,\boldsymbol \xi = \boldsymbol \xi(\true, \boldsymbol \lambda, \boldsymbol y)$, belonging to the segment between  $\boldsymbol \lambda$ and $\true$, that the functions $\partial L/\partial \boldsymbol \lambda$ and $\partial L/\partial \boldsymbol \alpha$ can be written as: 
\begin{align}\label{TaylorS}
{\partial L}/{\partial \boldsymbol \lambda} &= \Delta[\true{}^{-1}] \boldsymbol y  - \boldsymbol 1 - \Delta[\true{}^{-2}] \cdot \Delta[\boldsymbol y]\cdot(\boldsymbol \lambda - \true)  + \Delta[\boldsymbol \lambda - \true] \cdot \Delta[\boldsymbol \phi^{-3}] \cdot \Delta[\boldsymbol y]\cdot(\boldsymbol \lambda - \true)\nonumber \\
&+ \Delta[\true{}^{-1}] \cdot\mathbf{D}'\boldsymbol \alpha \, - \,  \Delta[\,\, \Delta[\boldsymbol \psi^{-2}] \mathbf{D}'\boldsymbol \alpha \,\,]\cdot (\boldsymbol \lambda - \true), \\[10pt]
{\partial L}/{\partial \boldsymbol \alpha} &= 
\mathbf{D}\Delta[\true{}^{-1}]\cdot(\boldsymbol \lambda - \true) -\frac{1}{2} \mathbf{D}\Delta[\boldsymbol \xi^{-2}] \cdot \Delta[\boldsymbol \lambda - \true]\cdot(\boldsymbol \lambda - \true). \nonumber 
\end{align}
Rewrite the  non-stochastic version of (\ref{likeEqf}) using (\ref{TaylorS}), and then multiply both sides of the first 
equation from the left by $\Delta[\true{}^{1/2}]$. As the latter is a non-singular matrix, the solutions for the new system, 
\begin{align}\label{system1}
& \Delta[\true{}^{-1/2}] (\boldsymbol y - \true) -
\Delta[\true{}^{-3/2}] \cdot \Delta[\boldsymbol y]\cdot({\boldsymbol \lambda} - \true) \nonumber\\
&+\Delta[\true{}^{1/2}] \cdot \Delta[\boldsymbol \lambda - \true] \cdot \Delta[\boldsymbol \phi^{-3}] \cdot \Delta[\boldsymbol y] \cdot (\boldsymbol \lambda - \true) \nonumber\\
&- \Delta[\true{}^{1/2}] \cdot\Delta[\,\,\Delta[\boldsymbol \psi^{-2}] \mathbf{D}'\boldsymbol \alpha\,\,] \cdot (\boldsymbol \lambda - \true) \nonumber\\
&+\Delta[\true{}^{-{1/2}}] \mathbf{D}'{\boldsymbol \alpha} = \boldsymbol 0, \\[10pt]
&\left\{\mathbf{D}\Delta[\true{}^{-1/2}] - {1}/{2} \mathbf{D}\Delta[\true{}^{1/2}] \cdot\Delta[\boldsymbol \xi^{-2}]\cdot\Delta[\boldsymbol \lambda - \true]\right\}\cdot \Delta[\true{}^{-1/2}]\cdot(\boldsymbol \lambda - \true) = \boldsymbol 0, \nonumber
\end{align}
are the same as those of (\ref{likeEqf}). After rearranging the terms in (\ref{system1}) and noticing that 
$$
\Delta[\true{}^{-1}] \cdot \Delta[\boldsymbol y] = \mathbf{I}+ \Delta[\true{}^{-1}] \cdot\Delta[\boldsymbol y - \true],$$ 
the non-stochastic version of the system (\ref{likeEqf}) can be written as:
\vspace{-4mm}
\begin{align}\label{system111}
& \nonumber \\ 
& \Delta[\true{}^{-1/2}] (\boldsymbol y - \true) -
\left \{ \mathbf{I}+ \Delta[\true{}^{-1}] \cdot\Delta[\boldsymbol y - \true]
-\Delta[\true]\,\,\Delta[\boldsymbol \lambda - \true]\,\,\Delta[\boldsymbol \phi^{-3}]\,\, \Delta[\boldsymbol y] \right. \nonumber\\
&\left.+ \Delta[\true] \,\,\, \Delta[\,\Delta[\boldsymbol \psi^{-2}] \mathbf{D}'\boldsymbol \alpha\,\,] \,\right\} \cdot \Delta[\true{}^{-1/2}](\boldsymbol \lambda - \true) 
+\Delta[\true{}^{-{1/2}}]\, \mathbf{D}'{\boldsymbol \alpha}= \boldsymbol 0, \\[10pt]
&\left\{\mathbf{D}\Delta[\true{}^{-1/2}] - {1}/{2} \mathbf{D}\Delta[\true{}^{1/2}]\cdot\Delta[\boldsymbol \xi^{-2}]\cdot\Delta[\boldsymbol \lambda - \true]\right\}\cdot\Delta[\true{}^{-1/2}](\boldsymbol \lambda - \true) = \boldsymbol 0.\nonumber  
\end{align}

\noindent Equivalently, (\ref{system111}) can be written in a matrix form:

\begin{equation}\label{systemPnewMatrix_estimate}
\left[\begin{array}{c} 
\Delta[\true{}^{-1/2}](\boldsymbol y - \true)\\
\boldsymbol 0
\end{array}\right] = \left[\begin{array}{rcr} \mathbf{I}+ c_1(\true, \boldsymbol \lambda, \boldsymbol y) &{ } & -\mathbf{H} \\ -\mathbf{H}' + c_2(\true, \boldsymbol \lambda)& { } & \boldsymbol 0\end{array}\right]\left[\begin{array}{c}\Delta[\true{}^{-1/2}]({\boldsymbol \lambda} - \true)\\
{\boldsymbol \alpha} \end{array}\right],
\end{equation}
where 
\begin{equation}\label{Hintens}
\mathbf{H} = \Delta[\true{}^{-{1/2}}]\mathbf{D}',
\end{equation} and
\begin{align}\label{c1c2}
&c_1(\true, \boldsymbol \lambda, \boldsymbol y) = \Delta[\true{}^{-1}]\cdot\Delta[\boldsymbol y - \true] 
-\Delta[\true] \cdot\Delta[\boldsymbol \lambda - \true] \cdot \Delta[\boldsymbol \phi^{-3}]\cdot \Delta[\boldsymbol y] \nonumber\\
&+ \Delta[\true]\cdot\Delta[\Delta[\boldsymbol \psi^{-2}] \mathbf{D}'\boldsymbol \alpha] \\[10pt]
&c_2(\true, \boldsymbol \lambda) = -\frac{1}{2} \mathbf{D}\Delta[\true{}^{1/2}]\cdot\Delta[\boldsymbol \xi^{-2}]\cdot\Delta[\boldsymbol \lambda - \true]. \nonumber
\end{align}

\vspace{1mm}

\noindent For every realization $\boldsymbol y  \in \mathscr{E}$, the MLE $\hat{\boldsymbol \lambda}= \hat{\boldsymbol \lambda}(\boldsymbol y)$ and the Lagrange multipliers $\hat{\boldsymbol \alpha}=\hat{\boldsymbol \alpha}(\boldsymbol y)$ are the unique solution to (\ref{systemPnewMatrix_estimate}). It will be shown next that $\,\,c_1(\true, \hat{\boldsymbol \lambda}, \boldsymbol y) = o(\boldsymbol 1)$ and $c_2(\true, \hat{\boldsymbol \lambda}) = o(\boldsymbol 1)$ when $\|\true\| \to \infty$. Notice first, that Lemma \ref{newlemma} ensures that Lemma \ref{smallLemma} can be used. The latter implies that ${\boldsymbol \phi}$, ${\boldsymbol \psi}$, and ${\boldsymbol \xi}$ in the Taylor expansion (\ref{TaylorS}) satisfy:
\begin{align}\label{PhisAsymp}
\Delta[\true{}^3]\cdot\Delta[{\boldsymbol \phi}{}^{-3}] &= \mathbf{I}+ O(\|\true\|^{-1/2}), \nonumber \\
\Delta[\true{}^{2}]\cdot\Delta[{\boldsymbol \psi}{}^{-2}] &=\mathbf{I}+ O(\|\true\|^{-1/2}), \\
\quad \Delta[\true{}^{2}]\cdot \Delta[{\boldsymbol \xi}{}^{-2}] &=\mathbf{I}+ O(\|\true\|^{-1/2}). \nonumber
\end{align}
The substitution of (\ref{PhisAsymp}) into (\ref{c1c2}) leads to:
\begin{align*}
&c_1(\true, \hat{\boldsymbol \lambda}, \boldsymbol y)  = \Delta[\true{}^{-1/2}]\cdot \Delta[\true{}^{-1/2}]\cdot\Delta[\boldsymbol y - \true] 
-\Delta[\true{}^{-2}] \cdot \Delta[\hat{\boldsymbol \lambda} - \true] \cdot \Delta[\true{}^{3}] \cdot \Delta[\hat{\boldsymbol \phi}{}^{-3}] \cdot \Delta[\boldsymbol y] \nonumber \\
&+ \Delta[\true{}^{-1}] \cdot\Delta[\,\Delta[\true{}^{2}] \cdot \Delta[\hat{\boldsymbol \psi}{}^{-2}] \mathbf{D}'\hat{\boldsymbol \alpha} \,], \\
&\nonumber \\
&c_2(\true, \hat{\boldsymbol \lambda}) = -\frac{1}{2} \mathbf{D}\Delta[\true{}^{-3/2}] \cdot \Delta[\true{}^{2}] \cdot \Delta[\hat{\boldsymbol \xi}{}^{-2}]\cdot\Delta[\hat{\boldsymbol \lambda} - \true]. \nonumber 
\end{align*}
After accounting for (\ref{Op}) and (\ref{majorY}), one obtains that 
\begin{align*}
&c_1(\true, \hat{\boldsymbol \lambda}, \boldsymbol y) = \Delta[\true{}^{-1/2}] \cdot O(\boldsymbol 1) + \Delta[\true{}^{-1}] \cdot\left(\mathbf{I}+ O(\|\true\|^{-1/2}) \mathbf{D}'\hat{\boldsymbol \alpha}\right) \\
& - \Delta[\true{}^{-1}] \cdot\Delta[\hat{\boldsymbol \lambda} - \true] \cdot \left(\mathbf{I}+ \Delta[\true{}^{-1}]\Delta[\boldsymbol y - \true]\right) \cdot \left(\mathbf{I}+ O(\|\true\|^{-1/2})\right), \\
&\nonumber \\
&c_2(\true, \hat{\boldsymbol \lambda}) = -\frac{1}{2} \mathbf{D}\Delta[\true{}^{-3/2}]\cdot \Delta\left[(\mathbf{I}+ O(\|\true\|^{-1/2}))(\hat{\boldsymbol \lambda} - \true)\right] \nonumber \\
&= -\frac{1}{2} \mathbf{D}\Delta[\true{}^{-1}] \cdot \Delta\left[(\mathbf{I}+ O(\|\true\|^{-1/2}))\cdot \Delta[\true{}^{-1/2}](\hat{\boldsymbol \lambda} - \true)\right],
\end{align*}  
which shows that,  when $\|\true\| \to \infty$, both $\,c_1(\true, \hat{\boldsymbol \lambda}, \boldsymbol y)$ and $c_2(\true, \hat{\boldsymbol \lambda})$ are of the order $o(\boldsymbol 1)$, and thus, can be assumed arbitrary small. The matrix 
$$
\left[\begin{array}{ccr} \mathbf{I}&{ } & -\mathbf{H} \\ -\mathbf{H}' & { } & \boldsymbol 0\end{array}\right]
$$
is non-singular, and $c_1(\true, \hat{\boldsymbol \lambda}, \boldsymbol y)$ and $c_2(\true, \hat{\boldsymbol \lambda})$ may be assumed small enough, so that the matrix of the system (\ref{systemPnewMatrix_estimate}) is also non-singular. In this case, 
\begin{eqnarray}\label{systemPnewMatrix_estimate_inverse}
\left[\begin{array}{c}\Delta[\true{}^{-1/2}](\hat{\boldsymbol \lambda} - \true)\\
\hat{\boldsymbol \alpha} \end{array}\right] &=&\left[\begin{array}{rcr} \mathbf{I}+ c_1(\true, \hat{\boldsymbol \lambda}, \boldsymbol y) &{ } & -\mathbf{H} \\ -\mathbf{H}' + c_2(\true, \hat{\boldsymbol \lambda})& { } & \boldsymbol 0\end{array}\right]^{-1}\left[\begin{array}{c} 
\Delta[\true{}^{-1/2}](\boldsymbol y - \true)\\
\boldsymbol 0
\end{array}\right] \nonumber \\
&&{} \nonumber \\
&=& \left[\begin{array}{ccr} \mathbf{I} &{ } & -\mathbf{H} \\ -\mathbf{H}'& { } & \boldsymbol 0\end{array}\right]^{-1}\left[\begin{array}{c} 
\Delta[\true{}^{-1/2}](\boldsymbol y - \true)\\
\boldsymbol 0
\end{array}\right] + o(\boldsymbol 1).
\end{eqnarray}
\noindent It can be verified directly that:
\begin{equation}\label{InverseMx}
\left[\begin{array}{ccr} \mathbf{I} &{ } & -\mathbf{H} \\ -\mathbf{H}'& { } & \boldsymbol 0\end{array}\right]^{-1} =  \left[\begin{array}{cc}\mathbf{I} - \mathbf{H}(\mathbf H' \mathbf{H})^{-1}\mathbf{H}' & - \mathbf{H}(\mathbf H' \mathbf{H})^{-1} \\ - (\mathbf H' \mathbf{H})^{-1}\mathbf{H}'&  - (\mathbf H' \mathbf{H})^{-1}\end{array}\right]. 
\end{equation}
The formula in the second line of (\ref{systemPnewMatrix_estimate_inverse}) provides an approximation to the right hand side of the non-stochastic version of (\ref{likeEqf}), and now the solution to (\ref{systemPnewMatrix_estimate_inverse}) will be studied.  First, one derives from (\ref{systemPnewMatrix_estimate_inverse}), using (\ref{Hintens}), that 
\begin{align}\label{AsympNS}
\Delta[\true{}^{-1/2}](\hat{\boldsymbol \lambda} - \true) \quad = &\quad {\mathbf{R}}^*\Delta[\true{}^{-1/2}](\boldsymbol y - \true) + o(\boldsymbol 1),
\end{align} 
where
\vspace{-2mm}
\begin{align*}
{\mathbf{R}}^* &= \mathbf{I} - \Delta[\true{}^{-{1/2}}]\mathbf{D}'\left(\mathbf{D}\Delta[\true{}^{-{1}}]\mathbf{D}'\right)^{-1} \mathbf{D} \Delta[\true{}^{-{1/2}}].
\end{align*} 
It is verified next that ${\mathbf{R}}^* = \mathbf{I} - \mathbf{R} + o(\boldsymbol 1)$, where $\mathbf{R} = \mathbf{D}'\left(\mathbf{D} \mathbf{D}'\right)^{-1} \mathbf{D}$, and thus (\ref{AsympNS}) can be simplified:
\vspace{-2mm}
\begin{align*}
{\mathbf{R}^*} &= \mathbf{I} - \Delta[\true{}^{-{1/2}}]\mathbf{D}'\left(\mathbf{D}\Delta[\true{}^{-{1}}]\mathbf{D}'\right)^{-1} \mathbf{D} \Delta[\true{}^{-{1/2}}] 
\\
&= \mathbf{I} -  O(\|\true\|^{-{1/2}})
\cdot \mathbf{D}'\left(\mathbf{D} \cdot O(\|\true\|^{-{1}}) \cdot \mathbf{D}'\right)^{-1} \mathbf{D} \cdot O(\|\true\|^{-{1/2}}) \\
&=\mathbf{I} -  \mathbf{D}'\left(\mathbf{D} \mathbf{D}'\right)^{-1} \mathbf{D} \cdot O(\boldsymbol 1) = \mathbf{I} - \it{k} \mathbf{R} + o(\boldsymbol 1), \quad \mbox{ as } \|\true\| \to \infty,
\end{align*}
where $k$ is a constant. To show  that $\it{k} = 1$, observe first that ${\mathbf{R}^*}{\mathbf{R}}^*= {\mathbf{R}}^*$, and ${\mathbf{R}}{\mathbf{R}} = {\mathbf{R}}$. Hence,
\begin{align*}
{\mathbf{R}^*}{\mathbf{R}}^*\equiv (\mathbf{I} - \it{k} \mathbf{R})^2 = \mathbf{I} - 2 \it{k} \mathbf{R} +  \it{k}^2 \mathbf{R}^2 =  \mathbf{I} - 2 \it{k} \mathbf{R} +  \it{k}^2 \mathbf{R} =  \mathbf{I} + (\it{k}^2 - 2 \it{k} )\mathbf{R} \equiv \mathbf{I} - \it{k} \mathbf{R},
\end{align*}
which entails that $\it{k}^2 - 2 \it{k}  = -\it{k}$, or $\it{k} (\it{k} - 1) = 0$, and thus, as $\it{k} \neq 0$, $\it{k} = 1$. Therefore, (\ref{AsympNS}) can be rewritten as: 
\begin{align}\label{AsympNewNS}
\Delta[\true{}^{-1/2}](\hat{\boldsymbol \lambda} - \true) \quad = &\quad (\mathbf{I}-{\mathbf{R}})\Delta[\true{}^{-1/2}](\boldsymbol y - \true) + o(\boldsymbol 1).
\end{align} 
The result in (\ref{AsympNewNS}) applies to the solution $\hat{\boldsymbol \lambda}$ of (\ref{systemPnewMatrix_estimate_inverse}), which approximated the non-stochastic version of (\ref{likeEqf}).  But, because $\Delta[\true{}^{-1/2}](\boldsymbol Y - \true) = O_p(\boldsymbol 1)$, (\ref{AsympNewNS}) is also true in the stochastic sense: 
\begin{align}\label{Asymp}
\Delta[\true{}^{-1/2}](\hat{\boldsymbol \lambda} - \true) \quad = &\quad (\mathbf{I} - \mathbf{R})\Delta[\true{}^{-1/2}](\boldsymbol Y - \true) + o_p(\boldsymbol 1).
\end{align} 
By Lemma \ref{DataNormala}, $\Delta[\true{}^{-1/2}]({\boldsymbol Y} - \true)$ is asymptotically, as $\|\true\| \to \infty$,  normal with zero mean and covariance matrix $\mathbf{I}$. Therefore, using the delta-method, conditionally on $\boldsymbol Y \in \mathscr{E}$, $\Delta[\true{}^{-1/2}]({\hat{\boldsymbol \lambda}} - \true)$ is asymptotically normal with zero mean and covariance matrix equal to $(\mathbf{I}-\mathbf{R}) \mathbf{I}(\mathbf{I}- \mathbf{R}) = (\mathbf{I}-\mathbf{R})^2 = \mathbf{I}- \mathbf{R} $:
\begin{equation}\label{CondIntens}
\mathcal{L} \left(\Delta[\true{}^{-1/2}]({\hat{\boldsymbol \lambda}} - \true) \,\, {\boldsymbol{\mid}} \, \boldsymbol Y \in \mathscr{E} \right) \quad {\to} \quad \mathcal{N
}(\boldsymbol 0, \mathbf{I} -\mathbf{D}'(\mathbf{D}\mathbf{D}')^{-1}\mathbf{D}).
\end{equation}

If the assumption $\boldsymbol Y \in \mathscr{E}$ is not made and $\boldsymbol Y$ is allowed to have some zeros, the MLE may or may not exist. In this case, the augmented MLE  $\tilde{\boldsymbol{\lambda}}$ and Lagrange multipliers $\tilde{\boldsymbol \alpha}$, defined in (\ref{pseudoMLE_intens}), need to be used. Set $$\boldsymbol\xi^{(r)}\equiv \left[\begin{array}{c}\Delta[\true{}^{-1/2}](\tilde{\boldsymbol \lambda} - \true)\\
\tilde{\boldsymbol \alpha} \end{array}\right],$$
$$\boldsymbol\xi^{(r)}_1 \equiv \left[\begin{array}{c}\Delta[\true{}^{-1/2}](\hat{\boldsymbol \lambda} - \true)\\
\hat{\boldsymbol \alpha} \end{array}\right], \quad \boldsymbol\xi^{(r)}_2 \equiv  \left[\begin{array}{c} 
\Delta[\true{}^{-1/2}]((\boldsymbol 1'\boldsymbol Y/I)\boldsymbol 1- \true)\\
\boldsymbol 0
\end{array}\right].$$ 
Then, for an arbitrary  $\boldsymbol t \in \mathbb{R}^{I+K}$,
\begin{align}
\mathbb{P}(\boldsymbol\xi^{(r)} \leq \boldsymbol t) 
&= \mathbb{P}(\boldsymbol\xi^{(r)}_1 \leq \boldsymbol t \mid \boldsymbol Y \in \mathscr{E})\cdot\mathbb{P}(\boldsymbol Y \in\mathscr{E}) +  \mathbb{P}(\boldsymbol\xi^{(r)}_2\leq \boldsymbol t \mid \boldsymbol Y \in \bar{\mathscr{E}})\cdot\mathbb{P}(\boldsymbol Y \in \bar{\mathscr{E}}) \nonumber.
\end{align}
As $\|\true\| \to \infty$, the asymptotic distribution of $\boldsymbol\xi^{(r)}_1$, given that $\boldsymbol Y \in \mathscr{E}$, exists and can be derived from (\ref{CondIntens}).
Also notice that, because $Y_1, \dots, Y_I$ are independent, $$1 \geq \mathbb{P}(\boldsymbol Y \in \mathscr{E}) \geq \prod_i\mathbb{P}(Y_i > 0) = \prod_i (1-\exp(-\ti)) \geq (1- \exp (-\|\true\|))^I \to 1,\quad \mbox{as } \|\true\| \to \infty,$$
and therefore, $\,\mathbb{P}(\boldsymbol Y \in \mathscr{E}) \to 1$ and, respectively, $\mathbb{P}(\boldsymbol Y \in \bar{\mathscr{E}}) \to 0$, as  $\|\true\| \to \infty$. Finally, Corollary \ref{CorollaryIntensNormal}  implies that $\boldsymbol\xi^{(r)}_2 = O_p(\boldsymbol 1)$, and thus, $\mathbb{P}(\boldsymbol\xi^{(r)}_2\leq \boldsymbol t \mid \boldsymbol Y \in \bar{\mathscr{E}}) = O(1)$. Therefore, with probability tending to $1$, the cumulative distribution functions $\mathbb{P}(\boldsymbol\xi^{(r)} \leq \boldsymbol t)$ and $\mathbb{P}(\boldsymbol\xi^{(r)}_1 \leq \boldsymbol t \mid \boldsymbol Y \in \mathscr{E})$ have the same limit, which implies that, as $\|\true\| \to \infty$, the asymptotic distribution of $\boldsymbol\xi^{(r)}$ exists and coincides with the asymptotic distribution of $\boldsymbol\xi^{(r)}_1$, given $\boldsymbol Y \in \mathscr{E}$. Therefore, the asymptotic behaviour, established in (\ref{CondIntens}) for data for which the MLEs exist, will be the same for all sequences of realizations.
\end{proof}

\vspace{3mm}

The case of probabilities is considered next. Let $\ptrue = (\ptone, \dots, \ptI)' > \boldsymbol 0$,  such that $\sum_{i=1}^I \pti = 1$, and consider a sequence of random vectors $\{\boldsymbol Y_N\}_{N=1}^{\infty}$, such that for every $N \in \mathbb{N}$, $\boldsymbol Y_N=(Y_{N1}, \dots, Y_{NI})$ has a multinomial distribution: $\boldsymbol Y_N \sim Mult(N, \ptrue)$. The asymptotic results will be obtained under the standard assumption that $N \to \infty$.

\begin{lemma}\label{MultiN}
Let $\boldsymbol Y_N \sim Mult(N, \ptrue)$ and $\boldsymbol p_N = \frac{1}{N}\boldsymbol Y_N$. Then, as $N \to \infty$,
\begin{equation}\label{MultiNormal}
\mathcal{L}\left\{N^{1/2}({\boldsymbol p}_N - \ptrue)\right\} \quad {\to} \quad \mathcal{N}(\boldsymbol 0, \Delta[\ptrue] - \ptrue 
\cdot \ptrue').
\end{equation}
\end{lemma}
For a proof, see Theorem 14.3-4 in \cite*{BFH}.

\begin{corollary}\label{MultiO}
Let $\boldsymbol Y_N \sim Mult(N, \ptrue)$, and $\boldsymbol y_N$ be a realization of $\boldsymbol Y_N$. Then, as $N \to \infty$, 
\begin{align*}
N^{-1/2}({\boldsymbol Y}_N - N\ptrue) \quad &{=} \quad O_p(\boldsymbol 1),\\
N^{-1/2}({\boldsymbol y}_N - N\ptrue) \quad &{=} \quad O(\boldsymbol 1). \nonumber
\end{align*}
\end{corollary}

As in the case of intensities, the MLE under a relational model for probabilities may or may not exist \citep{KRextended}, and in order to deal with this situation, an augmented MLE and augmented Lagrange multipliers are defined. Let $\mathscr{E}_N$ denote the set of values of $\boldsymbol Y_N$, for which the MLE exist, and $\bar{\mathscr{E}}_N$ denote its complement in $\mathbb{R}_{\geq 0}$. Set

\begin{equation}\label{pseudoMLE_prob}
\tilde{\boldsymbol p}_N =\left\{\begin{array}{ll}
 \hat{\boldsymbol p}_N,   & \mbox{ if } \boldsymbol Y_N \in \mathscr{E}_N, 
 \\
 (1/I)\boldsymbol 1,  & \mbox{ if } \boldsymbol Y_N \in \bar{\mathscr{E}}_N, 
 \end{array} \right.   \qquad (\tilde{\alpha_0}, \tilde{\boldsymbol \alpha}') =\left\{\begin{array}{ll}
 (\hat{\alpha}_0,  \hat{\boldsymbol \alpha}') & \mbox{ if } \boldsymbol Y_N \in \mathscr{E}_N, 
 \\
 (0, \boldsymbol 0),  & \mbox{ if } \boldsymbol Y_N \in \bar{\mathscr{E}}_N. 
 \end{array} \right.   
\end{equation} 
Notice that the normalization, $\boldsymbol 1'\tilde{\boldsymbol p}_N = 1$, holds.

\begin{theorem}\label{MLENormalProba1} 
Let $RM_{\boldsymbol p}(\mathbf{A})$ be a relational model for probabilities, $\mathbf{D}$ be a kernel basis matrix, $\ptrue \in RM_{\boldsymbol p}(\mathbf{A})$, and let $\boldsymbol Y_N \sim Mult(N, \ptrue)$ be the observations. Let $\hat{\boldsymbol p}_N$ be the MLE of $\ptrue$ under $RM_{\boldsymbol p}(\mathbf{A})$  and $\tilde{\boldsymbol p}_N$ be the augmented MLE. Then, as $N \to \infty$,
\begin{enumerate}[(i)]
\item $\mathcal{L} \left(N^{1/2}(\hat{\boldsymbol p}_N - \ptrue) \,\, {\boldsymbol{\mid}} \, \boldsymbol Y \in \mathscr{E} \right) \quad {\to}  \quad \mathcal{N
}(\boldsymbol 0, \mathbf{M}\boldsymbol \Sigma \mathbf{M}'),$\\
\item $\mathcal{L} \left(N^{1/2}({\tilde{\boldsymbol p}}_N - \ptrue) \right) \quad {\to} \quad\mathcal{N
}(\boldsymbol 0, \mathbf{M}\boldsymbol \Sigma \mathbf{M}').
$ 
\end{enumerate}
where 
\vspace{-4mm}
\begin{eqnarray}\label{Sigma}
\boldsymbol \Sigma &=& \Delta[\ptrue] - \ptrue \cdot \ptrue' \nonumber \\
\mathbf{M} &=& \mathbf{I} - \Delta[\ptrue] \cdot \mathbf{H}(\mathbf{H}'
\Delta[\ptrue]\mathbf{H})^{-1}\mathbf{H}',\\
\mathbf{H} &=&  \left(\boldsymbol 1, \Delta[\ptrue{}^{-1}]\mathbf{D}'\right).\nonumber
\end{eqnarray}

\end{theorem}

\noindent Thereafter, to simplify the notation, $\boldsymbol Y \equiv \boldsymbol Y_N$, $\hat{\boldsymbol p} \equiv \hat{\boldsymbol p}_N$, $\tilde{\boldsymbol p} \equiv \tilde{\boldsymbol p}_N$, $\mathscr{E} \equiv \mathscr{E}_N$.


\begin{proof}

The relational model, with or without the overall effect, can be expressed as: 
\begin{eqnarray} \label{zetaconst}
\boldsymbol 1' \boldsymbol p -1 = 0,  \qquad \mathbf{D}\mbox{log } \boldsymbol p = \boldsymbol 0. 
\end{eqnarray}
If $\boldsymbol Y \in \mathscr{E}$, that is the MLE $\hat{\boldsymbol p}$ under the model (\ref{zetaconst}), exists. In this case,  $\hat{\boldsymbol p}$ is the unique point that maximizes the Langrangian:
\begin{equation}\label{ASlagr}
{L}(\boldsymbol p, \alpha_0, \boldsymbol{\alpha}, \boldsymbol Y) = \boldsymbol y' \mbox{log } \boldsymbol p - N\mbox{log }(\boldsymbol 1' \boldsymbol p) + \alpha_0(\boldsymbol 1' \boldsymbol p - 1) + \boldsymbol{\alpha}'\mathbf{D}\mbox{log } \boldsymbol{p},
\end{equation}
over the interior of the $I-1$ dimensional simplex $\{\boldsymbol p \geq \boldsymbol 0: \,\, \sum_{i =1}^I = 1\}$.  
Here $\alpha_0$ and $\boldsymbol \alpha = (\alpha_1, \dots, \alpha_{K})'$ are the Lagrange multipliers. The MLE is therefore the unique solution to the system:
\begin{align}\label{likeEqf_M}
&{\partial L}/{\partial \boldsymbol p} = \Delta[{\boldsymbol p}^{-1}]\boldsymbol Y - {N}/(\boldsymbol 1' \boldsymbol p)\cdot\boldsymbol 1 +
\alpha_0 \cdot \boldsymbol 1 + \Delta[\boldsymbol {p}^{-1}] \mathbf{D}' \boldsymbol \alpha= \boldsymbol 0, \nonumber\\
&{\partial L}/{\partial \alpha_0} = \boldsymbol 1'\boldsymbol p - 1 = 0,\\
&{\partial L}/{\partial \boldsymbol \alpha} = \mathbf{D} \log\boldsymbol p = \boldsymbol 0. \nonumber
\end{align}
The remaining proof is similar to the one of Theorem \ref{MLENormala}.  Given $\boldsymbol Y =\boldsymbol y$,  expand the functions ${\partial L}/{\partial \boldsymbol p}$, ${\partial L}/{\partial \alpha_0}$, and ${\partial L}/{\partial \boldsymbol \alpha}$ into Taylor series around $\ptrue$, and substitute these expansions into the non-stochastic version of (\ref{likeEqf_M}) :
\begin{align}\label{system1_M}
& \Delta[\ptrue{}^{-1}] \boldsymbol y -
\Delta[\ptrue{}^{-2}] \cdot \Delta[\boldsymbol y]\cdot({\boldsymbol p} - \ptrue) 
+ \Delta[\boldsymbol p - \ptrue] \cdot \Delta[\boldsymbol \phi^{-3}] \cdot \Delta[\boldsymbol y] \cdot (\boldsymbol p - \ptrue)\nonumber \\
& - N \cdot \boldsymbol 1 +\alpha_0 \cdot \boldsymbol 1 +\Delta[\ptrue{}^{-{1}}] \mathbf{D}'{\boldsymbol \alpha} - \Delta[\,\,\Delta[\boldsymbol \psi^{-2}] \mathbf{D}'\boldsymbol \alpha\,\,] \cdot (\boldsymbol p - \ptrue) = \boldsymbol 0, \nonumber \\[10pt]
&\boldsymbol 1'(\boldsymbol p - \ptrue) = 0,\\[10pt]
&\left\{\mathbf{D}\Delta[\ptrue{}^{-1}] - \frac{1}{2} \mathbf{D}\cdot\Delta[\boldsymbol \xi^{-2}]\cdot\Delta[\boldsymbol p - \ptrue]\right\}\cdot(\boldsymbol p - \ptrue) = \boldsymbol 0. \nonumber 
\end{align}
Here $\boldsymbol \phi = \boldsymbol \phi(\ptrue, \boldsymbol p, \boldsymbol y)$, $\boldsymbol \psi = \boldsymbol \psi(\ptrue, \boldsymbol p, \boldsymbol y)$, $\boldsymbol \xi = \boldsymbol \xi(\ptrue, \boldsymbol p, \boldsymbol y)$ are on the segment between $\boldsymbol p$ and $\ptrue$. 
Multiply the first equation by $N^{-1/2}$, and the second and the third equations by $N^{1/2}$, and rearrange the terms in the system (\ref{system1_M}), then rewrite it in the matrix form:
\begin{eqnarray}\label{systemPnewMatrix_estimate_M}
& &\left[\begin{array}{c} 
N^{1/2}\Delta[\ptrue{}^{-1}](\boldsymbol y/N - \ptrue)\\
0\\
\boldsymbol 0
\end{array}\right] \\
&{}&\nonumber \\
& =& \left[\begin{array}{crc} \Delta[\ptrue{}^{-1}] + b_1(\ptrue, \boldsymbol p, \boldsymbol y) & { -\boldsymbol 1 } & -\Delta[\ptrue{}^{-1}]\mathbf{D}' \\
-\boldsymbol 1' & 0 & \boldsymbol 0\\
-\mathbf{D}\Delta[\ptrue{}^{-1}] + b_2(\ptrue, \boldsymbol p)& 0 & \boldsymbol 0\end{array}\right]\left[\begin{array}{c}N^{1/2}({\boldsymbol p} - \ptrue)\\
N^{-1/2}\alpha_0\\
N^{-1/2}{\boldsymbol \alpha} \end{array}\right], \nonumber
\end{eqnarray}
where 
\begin{align*}
b_1(\ptrue, \boldsymbol p, \boldsymbol y) &= \Delta[\ptrue{}^{-2}]\cdot\Delta[\boldsymbol y/N - \ptrue] 
-\Delta[\boldsymbol p - \ptrue] \cdot \Delta[\boldsymbol \phi^{-3}]\cdot \Delta[\boldsymbol y/N] \nonumber\\
&+ N^{-1}\cdot\Delta[\Delta[\boldsymbol \psi^{-2}] \mathbf{D}'\boldsymbol \alpha] \\[10pt]
b_2(\ptrue, \boldsymbol p) &= \frac{1}{2} \mathbf{D}\Delta[\ptrue]\cdot\Delta[\boldsymbol \xi^{-2}]\cdot\Delta[\boldsymbol p - \ptrue].
\end{align*} 
Notice that (\ref{systemPnewMatrix_estimate_M}), being equivalent to (\ref{system1_M}), is also a non-stochastic version of (\ref{likeEqf_M}).  It can be shown that $b_1(\ptrue, \hat{\boldsymbol p}, \boldsymbol Y) = o(\boldsymbol 1)$ and $b_1(\ptrue, \hat{\boldsymbol p}) = o(\boldsymbol 1)$, as $N \to \infty$. The proof is omitted, as it is similar to the one  for $c_1$ and $c_2$ in Lemma \ref{MLENormala}. Thus, from (\ref{systemPnewMatrix_estimate_M}) one obtains that  
\begin{equation}\label{systemPnewMatrix_estimate_M1}  
{\small \left[\begin{array}{c}N^{1/2}(\hat{\boldsymbol p} - \ptrue)\\
N^{-1/2}\hat{\alpha}_0\\
N^{-1/2}\hat{\boldsymbol \alpha} \end{array}\right] = \left[\begin{array}{crc} \Delta[\ptrue{}^{-1}] & { -\boldsymbol 1 } & -\Delta[\ptrue{}^{-1}]\mathbf{D}' \\
-\boldsymbol 1' & 0 & \boldsymbol 0\\
 -\mathbf{D}\Delta[\ptrue{}^{-1}] & {0 } &  \boldsymbol 0\end{array}\right]^{-1}
\left[\begin{array}{c} 
N^{1/2}\Delta[\ptrue{}^{-1}](\boldsymbol Y/N - \ptrue)\\
0\\
\boldsymbol 0
\end{array}\right] + o(\boldsymbol 1). }
\end{equation}
It can be verified directly that
{\small
 \begin{eqnarray}\label{InverseMxProb}
&=& \left[\begin{array}{crc} \Delta[\ptrue{}^{-1}] &  { -\boldsymbol 1 } & -\Delta[\ptrue{}^{-1}]\mathbf{D}' \\
-\boldsymbol 1' &  0 &  \boldsymbol 0\\
 -\mathbf{D}\Delta[\ptrue{}^{-1}] &0& \boldsymbol 0\end{array}\right]^{-1}  \nonumber  \\
 &{}& \nonumber \\[-1pt]
  &=& \left[\begin{array}{lr} \Delta[\ptrue](\mathbf{I} - \mathbf{H}(\mathbf H' \Delta[\ptrue]\mathbf{H})^{-1}\mathbf{H}'\Delta[\ptrue]) &  -\Delta[\ptrue]\mathbf{H}(\mathbf H' \Delta[\ptrue]\mathbf{H})^{-1} \\ 
 &\\
 -(\mathbf H' \Delta[\ptrue]\mathbf{H})^{-1}\mathbf{H}'\Delta[\ptrue] & - (\mathbf H' \Delta[\ptrue]\mathbf{H})^{-1}\end{array}\right],
\end{eqnarray} }
where
\begin{align}\label{Hprobs}
&\mathbf{H} = \left(\boldsymbol 1, \Delta[\ptrue{}^{-1}]\mathbf{D}'\right).
\end{align}
\noindent The further argument proceeds as in the case of intensities.  Using (\ref{InverseMxProb}), one obtains from (\ref{systemPnewMatrix_estimate_M1}) that
\begin{equation}\label{AsympM}
N^{1/2}(\hat{\boldsymbol p} - \ptrue) = N^{1/2}\mathbf{S}\Delta[\ptrue{}^{-1}](\boldsymbol Y/N - \ptrue) + o(\boldsymbol 1), \qquad \mbox{as } N \to \infty,
\end{equation}
where
$$\mathbf{S} = \Delta[\ptrue]\cdot(\mathbf{I}- \mathbf{H}(\mathbf{H}'
\Delta[\ptrue]\mathbf{H})^{-1}\mathbf{H}'\Delta[\ptrue]).$$
Equivalently,
\begin{equation}\label{AsympM1}
N^{1/2}(\hat{\boldsymbol p} - \ptrue) = N^{1/2}\mathbf{M}(\boldsymbol Y/N - \ptrue) + o(\boldsymbol 1), \qquad \mbox{as } N \to \infty,
\end{equation}
where
$$\mathbf{M} = \mathbf{S}\Delta[\ptrue{}^{-1}] = \Delta[\ptrue]\cdot(\mathbf{I}- \mathbf{H}(\mathbf{H}'
\Delta[\ptrue]\mathbf{H})^{-1}\mathbf{H}'\Delta[\ptrue]) \cdot\Delta[\ptrue{}^{-1}]$$
$$=\mathbf{I} - \Delta[\ptrue] \cdot \mathbf{H}(\mathbf{H}'
\Delta[\ptrue]\mathbf{H})^{-1}\mathbf{H}'.$$
The result in (\ref{AsympM1}) applies to the solution $\hat{\boldsymbol p}$ of (\ref{systemPnewMatrix_estimate_M1}), which approximates the non-stochastic version of (\ref{likeEqf_M}). By Corollary \ref{MultiO},   $N^{1/2}(\boldsymbol Y/N - \ptrue) = O_p(\boldsymbol 1)$,  and thus, (\ref{AsympM1}) also holds in the stochastic sense:
\begin{equation}\label{AsympM1st}
N^{1/2}(\hat{\boldsymbol p} - \ptrue) = N^{1/2}\mathbf{M}(\boldsymbol Y/N - \ptrue) + o_p(\boldsymbol 1), \qquad \mbox{as } N \to \infty.
\end{equation}
As, by Lemma \ref{MultiN}, $N^{1/2}(\boldsymbol Y/N - \ptrue)$ is asymptotically normal with zero mean and covariance matrix $\boldsymbol \Sigma = \Delta[\ptrue] - \ptrue \cdot \ptrue'$, the asymptotic distribution of  $N^{1/2}(\hat{\boldsymbol p} - \ptrue)$, conditionally on $\boldsymbol Y \in \mathscr{E}$, is also normal, namely:
\begin{equation}\label{NormalityCondProb}
\mathcal{L} \left(N^{1/2}(\hat{\boldsymbol p}_N - \ptrue) \,\, {\boldsymbol{\mid}} \, \boldsymbol Y \in \mathscr{E} \right) \quad {\to}  \quad \mathcal{N
}(\boldsymbol 0, \mathbf{M}\boldsymbol \Sigma \mathbf{M}').
\end{equation}
If the assumption $\boldsymbol Y \in \mathscr{E}$ is not made, the MLE may or may not exist. In this case, the augmented  MLE and Langrange multipliers, defined in (\ref{pseudoMLE_prob}), need to be used. Similarly to the case of intensities, set $$\boldsymbol\xi^{(r)}\equiv \left[\begin{array}{c}N^{1/2}(\tilde{\boldsymbol p} - \ptrue)\\
N^{-1/2}\tilde{\alpha}_0\\
N^{-1/2}\tilde{\boldsymbol \alpha} \end{array}\right],$$
$$\boldsymbol\xi^{(r)}_1 \equiv \left[\begin{array}{c}N^{1/2}(\hat{\boldsymbol p} - \ptrue)\\
N^{-1/2}\hat{\alpha}_0\\
N^{-1/2}\hat{\boldsymbol \alpha} \end{array}\right], \quad \boldsymbol\xi^{(r)}_2 \equiv  \left[\begin{array}{c} 
N^{1/2}  (\boldsymbol 1/I - \ptrue)\\
0\\
\boldsymbol 0
\end{array}\right].$$ 
For an arbitrary  $\boldsymbol t \in \mathbb{R}^{I+K+1}$,
\begin{align}\label{CDFprob}
\mathbb{P}(\boldsymbol\xi^{(r)} \leq \boldsymbol t) 
&= \mathbb{P}(\boldsymbol\xi^{(r)}_1 \leq \boldsymbol t \mid \boldsymbol Y \in \mathscr{E})\cdot\mathbb{P}(\boldsymbol Y \in \mathscr{E}) +  \mathbb{P}(\boldsymbol\xi^{(r)}_2\leq \boldsymbol t \mid \boldsymbol Y \in \bar{\mathscr{E}})\cdot\mathbb{P}(\boldsymbol Y \in \bar{\mathscr{E}}).
\end{align}
Equation (\ref{NormalityCondProb}) implies that, as $N \to \infty$, the asymptotic distribution of $\boldsymbol\xi^{(r)}_1$, given that $\boldsymbol Y \in \mathscr{E}$, exists.
Also notice that $$\mathbb{P}(\boldsymbol Y \in \bar{\mathscr{E}}) \leq \mathbb{P}(\exists i \in \mathcal{I}: \,\, Y_i = 0) = O(( 1-\|\ptrue\|)^N) \to 0,\quad \mbox{as } N\to \infty,$$
and therefore, $N^{1/2}\mathbb{P}(\boldsymbol Y \in \bar{\mathscr{E}})  \to 0$ and $\,\mathbb{P}(\boldsymbol Y \in \mathscr{E}) \to 1$. Finally, because $(\boldsymbol 1/I - \ptrue)= O_p(\boldsymbol 1)$, $\,\,\mathbb{P}(\boldsymbol\xi^{(r)}_2\leq \boldsymbol t \mid \boldsymbol Y \in \bar{\mathscr{E}}) \mathbb{P}(\boldsymbol Y \in \bar{\mathscr{E}})= o(1)$. Therefore, (\ref{CDFprob}) entails that with probability tending to $1$, the cumulative distribution functions $\mathbb{P}(\boldsymbol\xi^{(r)} \leq \boldsymbol t)$ and $\mathbb{P}(\boldsymbol\xi^{(r)}_1 \leq \boldsymbol t \mid \boldsymbol Y \in \mathscr{E})$ have the same limit, which implies that, as $N \to \infty$, the asymptotic distribution of $\boldsymbol\xi^{(r)}$ exists and coincides with the asymptotic distribution of $\boldsymbol\xi^{(r)}_1$, given $\boldsymbol Y \in \mathscr{E}$. Therefore, the asymptotic behaviour, established in  (\ref{NormalityCondProb}) for data for which the MLEs exist, will be the same for all sequences of realizations. 
\end{proof}

\begin{corollary}\label{CovarianceOE}
If $RM_{\boldsymbol p}(\mathbf{A})$ is a model with the overall effect, the covariance matrix of the asymptotic distribution is equal to
$$ \boldsymbol \Sigma - \mathbf{D}'(\mathbf{D} \Delta[\ptrue{}^{-1}] \mathbf{D}')^{-1}  \mathbf{D}.$$
\end{corollary} 

\noindent The proof is deferred to the Appendix. The formula may also be obtained using Eq.(4.6) in \cite{LangPoissMult}. 

\vspace{1mm}

In the next section, the asymptotic distributions of the test statistics are derived.

\section{Asymptotic distributions of the Pearson  and the Bregman statistics}\label{TestStatSection}

Let $X^2 (\boldsymbol u, \boldsymbol v) $ and $B(\boldsymbol u, \boldsymbol v) $ denote the Pearson and the Bregman statistics, respectively: 
\begin{align}
&X^2 (\boldsymbol u, \boldsymbol v) = \sum_{i \in \mathcal{I}} (u_i - v_i)^2/v_i  \label{PS}\\
&B(\boldsymbol u, \boldsymbol v) = 2\cdot\sum_{i \in \mathcal{I}}\{ u_i \log(u_i/v_i) - (u_i - v_i)\}, \label{BrS}
\end{align} 
Here $\boldsymbol u, \boldsymbol v > \boldsymbol 0$, but accepting the convention $0\log 0=0$ allows to extend the definitions to $\boldsymbol u \geq \boldsymbol 0$. It will be established that, whether or not a relational model includes the overall effect, the asymptotic distributions of these statistics, when $\boldsymbol u = \boldsymbol Y$  and $\boldsymbol v$ is either the conditional or augmented MLE,  are equivalent and equal to a chi-squared distribution.  Separate proofs will be given for the case of intensities and for the case of probabilities. Both proofs rely on the following lemmas (proved in the Appendix), showing the equivalence of these statistics for certain non-stochastic sequences.

\vspace{1mm}

\begin{lemma}\label{NonStochBregmanPearsona} 
Let $\boldsymbol t_r  \geq \boldsymbol 0$ and $\boldsymbol \lambda_r > \boldsymbol 0$ be such that 
$$
\Delta[\boldsymbol \lambda_r^{-1}](\boldsymbol t_r - \boldsymbol \lambda_r) = O(\|\boldsymbol \lambda_r\|^{-1/2}), \qquad \mbox{as } \|\boldsymbol \lambda_r\| \to \infty.
$$ 
Then
$$
B(\boldsymbol t_r, \boldsymbol \lambda_r) - X^2(\boldsymbol t_r, \boldsymbol \lambda_r) = o(1), \qquad \mbox{as } \|\boldsymbol \lambda_r\| \to \infty.
$$
\end{lemma}

\vspace{1mm}

\begin{lemma}\label{NonStochBregmanPearson2a}
Let $\boldsymbol t_r  \geq \boldsymbol 0$ and $\boldsymbol \kappa_r > \boldsymbol 0$ satisfy: 
\begin{align*}
&\Delta[\boldsymbol \lambda_r^{-1}](\boldsymbol t_r - \boldsymbol \lambda_r) = O(\|\boldsymbol \lambda_r\|^{-1/2}), \quad \mbox{ and } \quad
\Delta[\boldsymbol \lambda_r^{-1}](\boldsymbol \kappa_r - \boldsymbol \lambda_r) = O(\|\boldsymbol \lambda_r\|^{-1/2}), 
\end{align*}
as $\|\boldsymbol \lambda_r\| \to \infty$.
Then,
$$
B(\boldsymbol t_r, \boldsymbol \kappa_r) - X^2(\boldsymbol t_r, \boldsymbol \kappa_r) = o(1), \qquad \mbox{as } \|\boldsymbol \lambda_r\| \to \infty.
$$
\end{lemma}
 Let $\mathscr{E}$ denote the set of values of $\boldsymbol Y$, for which the MLE exist, and $\bar{\mathscr{E}}$ denote its complement in $\mathbb{R}_{\geq 0}$. Define the augmented Pearson and Bregman statistics as:
\begin{align}  \label{AugmPearsonBregmanStat} 
&X^2 (\boldsymbol Y, \tilde{\boldsymbol \lambda}) =  \left\{\begin{array}{ll}
X^2 (\boldsymbol Y, \hat{\boldsymbol \lambda}),   & \mbox{ if } \boldsymbol Y \in \mathscr{E}, 
 \\
\sum_{i \in \mathcal{I}}\frac{I}{{\boldsymbol 1'\boldsymbol Y}}(Y_i - ({\boldsymbol 1'\boldsymbol Y}/{I}))^2,  &\mbox{ if } \boldsymbol Y \in \bar{\mathscr{E}},
\end{array} \right.  \nonumber\\
&  \\
&B(\boldsymbol Y, \tilde{\boldsymbol \lambda}) =  \left\{\begin{array}{ll}
B (\boldsymbol Y, \hat{\boldsymbol \lambda}),   & \mbox{ if } \boldsymbol Y \in \mathscr{E}, 
 \\
2\sum_{i \in \mathcal{I}}\{Y_i \log ({Y_i}/{({\boldsymbol 1'\boldsymbol Y}/{I})}) - (Y_i - ({\boldsymbol 1'\boldsymbol Y}/{I}))\},  &\mbox{ if } \boldsymbol Y \in \bar{\mathscr{E}}.
\end{array} \right. \nonumber
\end{align}

\vspace{1mm}

\noindent The common asymptotic distribution of the Pearson and the Bregman statistics under a model for intensities is determined next.  For simplicity of notation, $\boldsymbol Y \equiv \boldsymbol Y^{r}$, $\true \equiv \boldsymbol \lambda^{r}$, and $\hat{\boldsymbol \lambda} \equiv \hat{\boldsymbol \lambda}{}^{r}$. 

\begin{theorem}\label{MainThPoissonaR}
Let $RM_{\boldsymbol \lambda}(\mathbf{A})$ be a relational model, $\true \in  RM_{\boldsymbol \lambda}(\mathbf{A})$,  and let $\boldsymbol Y \sim Pois(\true)$ be the observations. Let $\hat{\boldsymbol \lambda}$ be the MLE of $\true$ under $RM_{\boldsymbol \lambda}(\mathbf{A}) $, and $\tilde{\boldsymbol \lambda}$ be the augmented MLE. Let $K = dim(Ker(\mathbf{A}))$. Then, as  $\|\true\| \to \infty$, 
$$\mathcal{L}(X^2 (\boldsymbol Y, \hat{\boldsymbol \lambda}) \mid \boldsymbol Y \in \mathscr{E}),  \quad \mathcal{L} (B(\boldsymbol Y, \hat{\boldsymbol \lambda}) \mid \boldsymbol Y \in \mathscr{E}),   \quad \mathcal{L}(X^2 (\boldsymbol Y, \tilde{\boldsymbol \lambda})),  \quad 
 \mathcal{L} (B(\boldsymbol Y, \tilde{\boldsymbol \lambda}))   \quad \to \quad \chi^2_{K}.$$
\end{theorem}


\begin{proof} 
The asymptotic conditional and unconditional distributions of the Pearson statistic are obtained first. Because of the equivalence of the two statistics, the corresponding distributions of the Bregman statistics will be the same.

 Suppose $\boldsymbol Y \in \mathscr{E}$, and let $\boldsymbol E = (E_1, \dots, E_I)'$ denote the vector of Pearson residuals:
$$
E_i = \hat{\lambda}_i^{-1/2}(Y_i - \hat{\lambda}_i), \quad i = 1, \dots, I.
$$
It will be shown now that, conditionally on $\boldsymbol Y \in \mathscr{E}$, $\boldsymbol E$ has a multivariate normal distribution. 
Let $\boldsymbol Y = (y_1, \dots, y_I)$, and define:
$$
e_i = e_i(y_i, \hat{\lambda}_i) = \hat{\lambda}_i^{-1/2}(y_i - \hat{\lambda}_i), \quad i = 1, \dots, I,
$$
Expand $e_i$ in a Taylor series around $\ti$, using the partial derivatives:
\begin{align*}
&\partial e_i / \partial y_i = \hat{\lambda}{}_i^{-1/2}, \quad \partial e_i / \partial \hat{\lambda}_i = -{1}/{2}(y_i + \hat{\lambda}_i)\hat{\lambda}_i^{-3/2},\\
&\partial e_i / \partial y_j = \partial e_i / \partial \hat{\lambda}_j = 0, \qquad i \neq j.
\end{align*}
As $e_i(\ti, \ti) = 0$, the Taylor expansion is
$$e_i(y_i, \hat{\lambda}_i) = \ti{}^{-1/2}(y_i - \ti) - \ti{}^{-1/2}(\hat{\lambda}_i - \ti) + o\left(\ti{}^{-1/2}|y_i - \ti|\right)+  o\left(\ti{}^{-1/2}|\hat{\lambda}_i - \ti|\right),$$
or, written in the vector form and using (\ref{AsympNewNS}),
\begin{align}\label{e2}
\boldsymbol e &=\Delta[\true{}^{-1/2}](\boldsymbol y - \true) - \Delta[\true{}^{-1/2}](\hat{\boldsymbol \lambda} - \true) \nonumber \\
                   &+ o(\|\Delta[\true{}^{-1/2}](\boldsymbol y - \true)\|) + o(\|\Delta[\true{}^{-1/2}](\hat{\boldsymbol \lambda} - \true)\|) \nonumber\\[10pt]
                   &=\Delta[\true{}^{-1/2}](\boldsymbol y - \true) - (\mathbf{I}-\mathbf{R})\Delta[\true{}^{-1/2}](\boldsymbol y - \true) + o(\boldsymbol 1)\nonumber \\
                   &+ o(\|\Delta[\true{}^{-1/2}](\boldsymbol y - \true)\|) + o\left(\|\Delta[\true{}^{-1/2}](\hat{\boldsymbol \lambda} - \true)\|\right) \nonumber\\[10pt]
                  &= \mathbf{R}\Delta[\true{}^{-1/2}](\boldsymbol y - \true) +  o(\boldsymbol 1). 
\end{align}
Because, by Corollary \ref{CorollaryIntensNormal}, $\Delta[\true{}^{-1/2}](\boldsymbol Y - \true)= O_p(\boldsymbol 1)$, (\ref{e2}) also holds in the stochastic sense:
$$\, \boldsymbol E = \mathbf{R}
\Delta[\true{}^{-1/2}](\boldsymbol Y - \true) +  o_p(\boldsymbol 1), \quad \mbox{ given that } \boldsymbol Y \in \mathscr{E}.$$ 
The latter implies that, conditionally on $\boldsymbol Y \in \mathscr{E}$,  $\boldsymbol E$ is asymptotically normal with zero mean and the covariance matrix $\mathbf{R}$. Notice further that, because $\mathbf{R}^2 =  (\mathbf{D}(\mathbf{D}'\mathbf{D})^{-1}\mathbf{D}')^2 = (\mathbf{D}(\mathbf{D}'\mathbf{D})^{-1}\mathbf{D}'= \mathbf{R}$,  the matrix $\mathbf{R}$ is idempotent, and therefore,
$$rank(\mathbf{R}) = trace(\mathbf{R}) = trace(\mathbf{D}(\mathbf{D}'\mathbf{D})^{-1}\mathbf{D}').$$
Because $\mathbf{D}(\mathbf{D}'\mathbf{D})^{-1}\mathbf{D}'$ is the projection matrix on the subspace spanned by the columns of $\mathbf{D}$, the rank of $\mathbf{D}$ is equal to the dimension of this subspace. Thus,  $$rank(\mathbf{R}) = trace(\mathbf{D}(\mathbf{D}'\mathbf{D})^{-1}\mathbf{D}') = rank(\mathbf{D}) = K.$$
Finally, because the Pearson statistic $X^2(\boldsymbol Y, \hat{\boldsymbol \lambda}) = \boldsymbol E'\boldsymbol E$, its asymptotic distribution, given that $\boldsymbol Y \in \mathscr{E}$, is chi-squared with $K$ degrees of freedom.

Suppose that $\boldsymbol Y \in \bar{\mathscr{E}}$. In this case,
\begin{align}
X^2(\boldsymbol Y, \tilde{\boldsymbol \lambda}) &= X^2(\boldsymbol Y, ({\boldsymbol 1'\boldsymbol Y}/{I}) = \sum_{i \in \mathcal{I}}\frac{I}{{\boldsymbol 1'\boldsymbol Y}}(Y_i - ({\boldsymbol 1'\boldsymbol Y}/{I}))^2. 
\end{align} 
By Corollary \ref{CorollaryIntensNormal}, $\Delta[\true{}^{-1/2}](\boldsymbol Y - \true)= O_p(\boldsymbol 1)$, and thus, ${\boldsymbol 1'\boldsymbol Y} = \boldsymbol 1'\true + O_p(\|\true{}\|^{1/2})$. Therefore, 
\begin{align}\label{BoundIntens}
X^2(\boldsymbol Y, ({\boldsymbol 1'\boldsymbol Y}/{I}) &= \frac{I}{\boldsymbol 1'\true + O_p(\|\true{}\|^{1/2})}\sum_{i \in \mathcal{I}}(Y_i - \ti - (({\boldsymbol 1'\boldsymbol Y}/{I}) - \ti))^2 \\
&= \frac{I}{\boldsymbol 1'\true + O_p(\|\true{}\|^{1/2})}O_p(\|\true{}\|) = \frac{O_p(\|\true{}\|)}{O_p(\|\true{}\|)(1 + O_p(\|\true{}\|^{-1/2})} = O_p(1). \nonumber
\end{align}
For an arbitrary  $t \geq 0$:
\begin{align}\label{CDFchisqIntens}
\mathbb{P}(X^2(\boldsymbol Y, \tilde{\boldsymbol \lambda}) \leq t) 
&= \mathbb{P}(X^2(\boldsymbol Y, \hat{\boldsymbol \lambda}) \leq t \mid \boldsymbol Y \in \mathscr{E})\cdot\mathbb{P}(\boldsymbol Y \in \mathscr{E})  \\
&+  \mathbb{P}(X^2(\boldsymbol Y, ({\boldsymbol 1'\boldsymbol Y}/{I})) \leq t \mid \boldsymbol Y \in \bar{\mathscr{E}})\cdot\mathbb{P}(\boldsymbol Y \in \bar{\mathscr{E}}). \nonumber
\end{align}
It was shown earlier in the proof that the asymptotic distribution of $X^2(\boldsymbol Y, \hat{\boldsymbol \lambda})$ given $\boldsymbol Y \in \mathscr{E}$ exists. Because $X^2(\boldsymbol Y, ({\boldsymbol 1'\boldsymbol Y}/{I})$ is bounded in probability, see (\ref{BoundIntens}), and, as it was shown in the proof of Theorem \ref{MLENormala}, $\mathbb{P}(\boldsymbol Y \in {\mathscr{E}}) \to 1$ and $\mathbb{P}(\boldsymbol Y \in \bar{\mathscr{E}}) \to 0$, as $\|\true\| \to \infty$, the asymptotic distribution of $X^2(\boldsymbol Y, \tilde{\boldsymbol \lambda}) $ also exists and is the same as the asymptotic conditional distribution of $X^2(\boldsymbol Y, \hat{\boldsymbol \lambda})$ given $\boldsymbol Y \in \mathscr{E}$, namely, $\chi^2_K$.

Finally, the asymptotic distribution of the Bregman statistic for the conditional and augmented MLEs is obtained. 

Recall that by Corollary \ref{CorollaryIntensNormal}, $\Delta[\true{}^{-1/2}](\boldsymbol Y - \true)= O_p(\boldsymbol 1)$, as $\|\true\| \to \infty$ an let now  $\ddot{\boldsymbol \lambda}$ denote either $\tilde{\boldsymbol \lambda}$ or $\hat{\boldsymbol \lambda}$, and in the second case, all statements will be conditional on  $\boldsymbol Y \in \mathscr{E}$.
Because, by Theorem \ref{MLENormala}, $\Delta[\true{}^{-1/2}](\ddot{\boldsymbol \lambda} - \true)$ is asymptotically normal,  it also holds that $\Delta[\true{}^{-1/2}](\ddot{\boldsymbol \lambda} - \true) = O_p(\boldsymbol 1)$.

Therefore, for $\boldsymbol t_r = \boldsymbol Y$, $\kappa_r = \ddot{\boldsymbol \lambda}$, and $\boldsymbol \lambda_r = \true$, one obtains that
\begin{equation} \label{IntensStochLemma}
\Delta[\boldsymbol \lambda_r^{-1}](\boldsymbol t_r - \boldsymbol \lambda_r) =  O_p(\|\boldsymbol \lambda_r\| ^{-1/2}) \, \mbox{ and } \, \Delta[\boldsymbol \lambda_r^{-1}](\boldsymbol \kappa_r - \boldsymbol \lambda_r) =  O_p(\|\boldsymbol \lambda_r\| ^{-1/2}).
\end{equation}
By Lemma   \ref{NonStochBregmanPearson2a}, applied to the non-stochastic version of (\ref {IntensStochLemma}): 
\begin{equation}\label{eqInt}
B(\boldsymbol y, \ddot{\boldsymbol \lambda}) -  X^2(\boldsymbol y, \ddot{\boldsymbol \lambda}) = o(1), \qquad \mbox{ as }  \|\true\| \to \infty,
\end{equation}
where $\boldsymbol y$ is a realizaton of $\boldsymbol Y$. By the standard argument, used in the previous section, (\ref{eqInt}) holds in the stochastic sense too:
$$B(\boldsymbol Y, \ddot{\boldsymbol \lambda}) -  X^2(\boldsymbol Y, \ddot{\boldsymbol \lambda}) = o_p(1), \qquad \mbox{ as }  \|\true\| \to \infty.$$
which completes the proof.
\end{proof}

The case of probabilities is considered next.   Let $\mathscr{E}_N$ denote the set of values of $\boldsymbol Y$, for which the MLE exist, and $\bar{\mathscr{E}}_N$ denote its complement in $\mathbb{R}_{\geq 0}$. Define the augmented Pearson and Bregman statistics as:
\begin{align}  \label{AugmPearsonBregmanStatProbs} 
&X^2 (\boldsymbol Y_N, N\tilde{\boldsymbol p}) =  \left\{\begin{array}{ll}
X^2 (\boldsymbol Y_N, N\hat{\boldsymbol p}),   & \mbox{ if } \boldsymbol Y_N \in \mathscr{E}_N, 
 \\
\sum_{i \in \mathcal{I}}\frac{I}{N} (Y_{Ni} - N/I)^2,  &\mbox{ if } \boldsymbol Y_N \in \bar{\mathscr{E}}_N,
\end{array} \right.  \nonumber\\
&  \\
&B(\boldsymbol Y_N, N\tilde{\boldsymbol p}) =  \left\{\begin{array}{ll}
B (\boldsymbol Y_N, N\hat{\boldsymbol p}),   & \mbox{ if } \boldsymbol Y_N \in \mathscr{E}_N, 
 \\
2\sum_{i \in \mathcal{I}}\{Y_{Ni} \log (I\cdot{Y_i}/N) - (Y_i - N/I)\},  &\mbox{ if } \boldsymbol Y_N \in \bar{\mathscr{E}}_N.
\end{array} \right. \nonumber
\end{align}
\noindent Thereafter, to simplify the notation, $\boldsymbol Y \equiv \boldsymbol Y_N$, $\hat{\boldsymbol p} \equiv \hat{\boldsymbol p}_N$, $\tilde{\boldsymbol p} \equiv \tilde{\boldsymbol p}_N$, $\mathscr{E} \equiv \mathscr{E}_N$.

\begin{theorem}\label{MainThMulta}
Let $RM_{\boldsymbol p}(\mathbf{A})$ be a relational model for probabilities, $\ptrue \in RM_{\boldsymbol p}(\mathbf{A})$, and let $\boldsymbol{Y} \sim Mult(N, \ptrue)$ be the observations.  Let $\hat{\boldsymbol p}$ be the MLE of $\ptrue$ under $RM_{\boldsymbol p}(\mathbf{A})$, and  $\tilde{\boldsymbol p}$ be the augmented MLE. Let $K = dim(Ker(\mathbf{A}))$. Then, as $N \to \infty$, 
$$\mathcal{L}(X^2 (\boldsymbol Y, N\hat{\boldsymbol p}) \mid \boldsymbol Y\in \mathscr{E}),   \,  \,  \mathcal{L} (B(\boldsymbol Y, N\hat{\boldsymbol p}) \mid \boldsymbol Y \in \mathscr{E}),    \, \,   \mathcal{L}(X^2 (\boldsymbol Y, N\tilde{\boldsymbol p})),   \, \, 
 \mathcal{L} (B(\boldsymbol Y, N\tilde{\boldsymbol p}))   \, \to \, \chi^2_{K}.$$
\end{theorem}

\begin{proof}

It is proved first that, in the case when the MLE exists, the Pearson residuals are normally distributed, and their covariance matrix and its rank are computed. This implies the asymptotic chi-squared distribution of the Pearson statistic and its number degrees of freedom. Second, the same result is being shown to hold in the augmented case, that is, when the MLE may or may not exist. Finally, the proof is completed by showing that the Pearson and Bregman statistics have the same asymptotic distribution.

Suppose $\boldsymbol Y \in \mathscr{E}$, and let $\boldsymbol E  = (E_1, \dots, E_I)'$ be the vector of Pearson residuals $E_i = (N\hat{p}_i)^{-1/2}(Y_i - N\hat{p}_i), \quad i = 1, \dots, I.$ 
It will now be shown that, conditionally on $\boldsymbol Y \in \mathscr{E}$, the distribution of $\boldsymbol E$ is multivariate normal.  

Let $\boldsymbol y$ be a realization of $\boldsymbol Y$, $\boldsymbol q = \boldsymbol y/N$, and 
$\boldsymbol e = (e_1, \dots, e_I)$, where $$
e_i = e_i(q_i, \hat{p}_i) = N^{1/2}\hat{p}_i^{-1/2}(q_i - \hat{p}_i), \quad i = 1, \dots, I.
$$
Compute the partial derivatives of $e_i(q_i, \hat{p}_i) $ with respect to $q_i$ and $\hat{p}_i$: 
\begin{align*}
&\partial e_i / \partial q_i = N^{1/2}\hat{p}_i^{-1/2}, \quad \partial e_i / \partial \hat{p}_i = - {N^{1/2}}(q_i + \hat{p}_i)\hat{p}_i^{-3/2}/2,\\
&\partial e_i / \partial q_j = \partial e_i / \partial \hat{p}_j = 0, \qquad i \neq j. 
\end{align*}
Notice that  $$e_i(\pti, \pti) = 0, \quad \partial e_i / \partial q_i (\pti, \pti) = N^{1/2}\pti{}^{-1/2}, \quad \partial e_i / \partial \hat{p}_i (\pti, \pti) = N^{1/2}\pti{}^{-1/2}.$$
The Taylor expansion of $e_i(q_i, \hat{p}_i) $ around $\pti$ is:
$$e_i(q_i, \hat{p}_i) = N^{1/2}\pti{}^{-1/2}(q_i - \pti) - N^{1/2}\pti{}^{-1/2}(\hat{p}_i - \pti) + o\left(|q_i- \pti|\right)+  o\left(|\hat{p}_i - \pti|\right).$$
Therefore, using (\ref{AsympM1}),
\begin{align}
\boldsymbol e &= N^{1/2}\Delta[\ptrue{}^{-1/2}](\boldsymbol q - \ptrue) - N^{1/2}\Delta[\ptrue{}^{-1/2}](\hat{\boldsymbol p} - \ptrue) + o(\|\boldsymbol q - \ptrue)\|) + o(\|\hat{\boldsymbol p} - \ptrue\|) \nonumber\\
                   {} & \nonumber \\
                   &=N^{1/2}\Delta[\ptrue{}^{-1/2}](\boldsymbol y/N - \ptrue) - \Delta[\ptrue{}^{-1/2}]\mathbf{S}\Delta[\ptrue{}^{-1}]N^{1/2}(\boldsymbol Y/N - \ptrue)+ o(\boldsymbol 1)\nonumber \\
                   &+ o(\|\Delta[\ptrue{}^{-1/2}](\boldsymbol y - \ptrue)\|) + o\left(\|\Delta[\ptrue{}^{-1/2}](\hat{\boldsymbol \lambda} - \ptrue)\|\right) \nonumber\\
                   {} & \nonumber \\
                    &=N^{1/2}\Delta[\ptrue{}^{-1/2}](\boldsymbol y/N - \ptrue) - \Delta[\ptrue{}^{-1/2}]\mathbf{S}\Delta[\ptrue{}^{-1/2}]N^{1/2}\Delta[\ptrue{}^{-1/2}](\boldsymbol Y/N - \ptrue)+ o(\boldsymbol 1)\nonumber \\
                   &+ o(\|\Delta[\ptrue{}^{-1/2}](\boldsymbol y - \ptrue)\|) + o\left(\|\Delta[\ptrue{}^{-1/2}](\hat{\boldsymbol \lambda} - \true)\|\right) \nonumber\\
                   {} & \nonumber \\
&= \left(\mathbf{I} -\Delta[\ptrue{}^{-1/2}]\mathbf{S}\Delta[\ptrue{}^{-1/2}]\right) \cdot N^{1/2}
\Delta[\ptrue{}^{-1/2}](\boldsymbol y/N - \ptrue) +  o(\boldsymbol 1), \qquad \mbox{as } N \to \infty.  \label{e2p}
\end{align}
As, by Corollary \ref{MultiO}, $N^{1/2}(\boldsymbol Y/N - \ptrue) =O_p(\boldsymbol 1)$,  (\ref{e2p}) also holds in the stochastic sense:
$$ \boldsymbol E = \left(\mathbf{I} -\Delta[\ptrue{}^{-1/2}]\mathbf{S}\Delta[\ptrue{}^{-1/2}]\right) \cdot N^{1/2}
\Delta[\ptrue{}^{-1/2}](\boldsymbol Y/N - \ptrue) +  o_p(\boldsymbol 1), \qquad \mbox{as } N \to \infty.$$
By Lemma \ref{MultiN}, $N^{1/2}(\boldsymbol Y/N - \ptrue)$ is asymptotically normal, with zero mean and covariance matrix $\boldsymbol \Sigma = \Delta[\ptrue] - \ptrue \cdot \ptrue'$. Therefore,  $N^{1/2}
\Delta[\ptrue{}^{-1/2}](\boldsymbol Y/N - \ptrue)$ is also asymptotically normal, with zero mean and covariance matrix equal to:
$$\boldsymbol \Sigma_1=\Delta[\ptrue{}^{-1/2}] \cdot\boldsymbol \Sigma \cdot \Delta[\ptrue{}^{-1/2}] = \mathbf{I} - \ptrue{}^{1/2} \cdot (\ptrue{}^{1/2})'.$$

Let $\mathbf{T} = \Delta[\ptrue{}^{-1/2}]\mathbf{S}\Delta[\ptrue{}^{-1/2}]$. The matrix $\mathbf{T}$ can be rewritten as 
\begin{align*}
\mathbf{T} &= \Delta[\ptrue{}^{-1/2}] \Delta[\ptrue] \left(\mathbf{I}- \mathbf{H}(\mathbf{H}'
\Delta[\ptrue] \mathbf{H})^{-1}\mathbf{H}'\Delta[\ptrue] \right)\Delta[\ptrue{}^{-1/2}] \\
&= \Delta[\ptrue{}^{1/2}]  \left(\mathbf{I}- \mathbf{H}(\mathbf{H}'
\Delta[\ptrue] \mathbf{H})^{-1}\mathbf{H}'\Delta[\ptrue] \right)\Delta[\ptrue{}^{-1/2}]\\
&= \mathbf{I} - \mathbf{C},
\end{align*}
where 
\begin{equation}\label{MatrixC}
\mathbf{C}= \Delta[\ptrue{}^{1/2}] \mathbf{H}(\mathbf{H}'
\Delta[\ptrue] \mathbf{H})^{-1}\mathbf{H}'\Delta[\ptrue{}^{1/2}].
\end{equation}
Therefore, given that $\boldsymbol Y \in \mathscr{E}$, the vector of Pearson residuals,
$$ \boldsymbol E = (\mathbf{I} - \mathbf{T})
\Delta[\ptrue{}^{-1/2}]N^{1/2}(\boldsymbol Y/N - \ptrue) + o_p(\boldsymbol 1) = \mathbf{C} \Delta[\ptrue{}^{-1/2}]N^{1/2}(\boldsymbol Y/N - \ptrue) + o_p(\boldsymbol 1),
$$
is asymptotically normal with zero mean and covariance matrix $\,\mathbf{C} \boldsymbol\Sigma_1 \mathbf{C}$:
\begin{equation}\label{PRnormProb}
\mathcal{L}(\boldsymbol E \mid \boldsymbol Y \in \mathscr{E} ) \quad \to \quad \mathcal{N}(\boldsymbol 0, \,\mathbf{C} \boldsymbol\Sigma_1 \mathbf{C}), \qquad \mbox{as } N \to \infty.
\end{equation}
\noindent The rank of the covariance matrix is computed next. Notice first that $\mathbf{C}$ is symmetric and $\mathbf{C}^2 = \mathbf{C}$. Therefore, $\mathbf{C}$ is a projection matrix, which, as seen from (\ref{MatrixC}), projects a vector on the subspace spanned by the columns of $\Delta[\ptrue{}^{1/2}] \mathbf{H}$. Also, because
\begin{align*}
\boldsymbol{\Sigma}_1^2 &= (\mathbf{I} - \ptrue{}^{1/2} \cdot (\ptrue{}^{1/2})')^2 
= \,\mathbf{I} - 2\ptrue{}^{1/2} \cdot (\ptrue{}^{1/2})' + \ptrue{}^{1/2} \cdot (\ptrue{}^{1/2})'\cdot\ptrue{}^{1/2} \cdot (\ptrue{}^{1/2})'=\\
&= \mathbf{I} - 2\ptrue{}^{1/2} \cdot (\ptrue{}^{1/2})' + \ptrue{}^{1/2} \cdot 1 \cdot (\ptrue{}^{1/2})'=\boldsymbol \Sigma_1,
\end{align*}
$\boldsymbol \Sigma_1$ is also a projection matrix.  

Both $\mathbf{C}$ and $\boldsymbol\Sigma_1$ are symmetric, and thus
$$\mathbf{C}\boldsymbol{\Sigma_1}\mathbf{C} = \mathbf{C}\boldsymbol \Sigma_1^2\mathbf{C} = (\mathbf{C}\boldsymbol 
\Sigma_1) (\mathbf{C}\boldsymbol \Sigma_1)'.$$
Using the properties of the matrix rank,
$$rank(\mathbf{C}\boldsymbol{\Sigma_1}\mathbf{C}) = rank((\mathbf{C}\boldsymbol 
\Sigma_1) (\mathbf{C}\boldsymbol \Sigma_1)') = rank(\mathbf{C}\boldsymbol 
\Sigma_1).$$
It will be shown now that $rank(\mathbf{C}\boldsymbol 
\Sigma_1) = dim(Image(\mathbf{C}\boldsymbol{\Sigma}_1) = K$. 

Any vector in $Image(\mathbf{C}\boldsymbol{\Sigma}_1)$ is obtained by two consecutive projections: first, by $\boldsymbol{\Sigma}_1$, and second, by $\mathbf{C}$. 
Because 
$$\boldsymbol{\Sigma}_1 \ptrue{}^{1/2} = (\mathbf{I} - \ptrue{}^{1/2} \cdot (\ptrue{}^{1/2})')\ptrue{}^{1/2} = \boldsymbol 0,$$
$\ptrue{}^{1/2} \in Ker(\boldsymbol \Sigma_1)$, and thus every vector in the image of  $\boldsymbol{\Sigma}_1$ is orthogonal to $\ptrue{}^{1/2}$.

On the other hand, $\mathbf{C}$ is a projection matrix on the set spanned by the columns of $
\Delta[\ptrue{}^{1/2}]\mathbf{H}$.
Since  $\Delta[\ptrue{}^{1/2}]\mathbf{H} = [\ptrue{}^{1/2}, \Delta[\ptrue{}^{-1/2}]\mathbf{D}']$, every vector in the image of $\mathbf{C}$ is a linear combination of $\ptrue{}^{1/2}$ and the columns of $\Delta[\ptrue{}^{-1/2}]\mathbf{D}'$.

Any vector, obtained as a projection by $\boldsymbol \Sigma_1$, is orthogonal to $\ptrue{}^{1/2}$. If this vector is projected by $\mathbf{C}$ on the span of $\ptrue{}^{1/2}$ and $\Delta[\ptrue{}^{-1/2}]\mathbf{D}'$, the coordinate corresponding to $\ptrue{}^{1/2}$ has to be equal to zero. Therefore, the image of $\mathbf{C}\boldsymbol \Sigma_1$ comprises all linear combinations of the columns of $\Delta[\ptrue{}^{-1/2}]\mathbf{D}'$ only. Because these columns are linearly independent and their number is exactly $K$,
$$rank(\mathbf{C}\boldsymbol{\Sigma}_1) = dim(Image(\mathbf{C}\boldsymbol{\Sigma}_1)) = K.$$
Therefore,  $rank(\mathbf{C} \boldsymbol\Sigma_1  \mathbf{C}) = K$.

Because for $\boldsymbol Y \in \mathscr{E}$, $\, X^2 (\boldsymbol Y, N\hat{\boldsymbol p})= \boldsymbol E' \boldsymbol E$,  where $\boldsymbol E$ is asymptotically normal with a covariance matrix of rank $K$, the asymptotic conditional distribution of the Pearson statistic, given $\boldsymbol Y \in \mathscr{E}$, is  $\boldsymbol \chi^2$ with $K$ degrees of freedom.

Second part of the proof deals with the case when the MLE is not supposed to exist. Assume, $\boldsymbol Y \in \bar{\mathscr{E}}$. Then $$X^2(\boldsymbol Y, N\tilde{\boldsymbol p}) = X^2(\boldsymbol Y, N{\boldsymbol 1}/I) = \sum_{i \in \mathcal{I}}{I}/{N} (Y_{i} - N/I)^2 = ( {I}/{N}\boldsymbol Y - \boldsymbol 1)'( {I}/{N}\boldsymbol Y - \boldsymbol 1).$$
Recall that, by Corollary \ref{MultiO}, $N^{1/2}(\boldsymbol Y/N - \ptrue) =O_p(\boldsymbol 1)$, as $N \to \infty$. Therefore,
\begin{align}\label{boundProb}
X^2(\boldsymbol Y, N{\boldsymbol 1}/I)  & = ( {I}/{N}\boldsymbol Y - \boldsymbol 1)'( {I}/{N}\boldsymbol Y - \boldsymbol 1) \nonumber \\
&= ({I}/{N}\boldsymbol Y - \ptrue - (\boldsymbol 1 - \ptrue))'({I}/{N}\boldsymbol Y - \ptrue - (\boldsymbol 1 - \ptrue)) \\
&= ( O_p(N^{-1/2}) -  (\boldsymbol 1 - \ptrue))' (O_p(N^{-1/2}) -  (\boldsymbol 1 - \ptrue)) = O_p(\boldsymbol 1), \nonumber
\end{align}
that is, for $\boldsymbol Y \in \bar{\mathscr{E}}$,  $X^2(\boldsymbol Y, N\tilde{\boldsymbol p}) = O_p(\boldsymbol 1)$, as $N \to \infty$.

For an arbitrary  $t \geq 0$,
\begin{align}\label{CDFchisqProb}
\mathbb{P}(X^2(\boldsymbol Y, N\tilde{\boldsymbol p}) \leq t) 
&= \mathbb{P}(X^2(\boldsymbol Y, N\hat{\boldsymbol p}) \leq t \mid \boldsymbol Y \in \mathscr{E})\cdot\mathbb{P}(\boldsymbol Y \in \mathscr{E}) \nonumber \\
&+  \mathbb{P}(X^2(\boldsymbol Y, N{\boldsymbol 1}/I) \leq t \mid \boldsymbol Y \in \bar{\mathscr{E}})\cdot\mathbb{P}(\boldsymbol Y \in \bar{\mathscr{E}}). 
\end{align}
It was shown above that, as $N \to \infty$, the asymptotic distribution of $X^2(\boldsymbol Y, N\hat{\boldsymbol p})$, given $\boldsymbol Y \in \mathscr{E}$, exists. Because for $\boldsymbol Y \in \bar{\mathscr{E}}$, $\,X^2(\boldsymbol Y, N\tilde{\boldsymbol p})$ is bounded in probability, see (\ref{boundProb}), and, as it was verified in the proof of Theorem \ref{MLENormalProba1}, $\mathbb{P}(\boldsymbol Y \in {\mathscr{E}}) \to 1$ and $\mathbb{P}(\boldsymbol Y \in \bar{\mathscr{E}}) \to 0$, when $N \to \infty$, the equation (\ref{CDFchisqProb}) implies that the asymptotic distribution of $X^2(\boldsymbol Y, N\tilde{\boldsymbol p})$ exists and is the same as the asymptotic conditional distribution of $X^2(\boldsymbol Y, N\hat{\boldsymbol p})$, given $\boldsymbol Y \in \mathscr{E}$, namely, $\chi^2_K$.

Finally, the asymptotic distribution of the Bregman statistic for the conditional and augmented MLEs is obtained.

By Corollary \ref{MultiO},  $N^{-1/2}({\boldsymbol Y} - N\ptrue) = O_p(\boldsymbol 1)$,  as $N \to \infty$.
Let $\ddot{\boldsymbol p}$ denote either $\tilde{\boldsymbol p}$ or $\hat{\boldsymbol p}$, in which case all statements will be conditional on $\boldsymbol Y \in \mathscr{E}$.
 By Theorem \ref{MLENormalProba1}, $N^{1/2}({\ddot{\boldsymbol p}} - \ptrue)$ is asymptotically normal, and therefore, $N^{1/2}({\ddot{\boldsymbol p}} - \ptrue) = O_p(\boldsymbol 1)$, or, equivalently, $N^{-1/2}({N\ddot{\boldsymbol p}} - N\ptrue) = O_p(\boldsymbol 1)$. Set  $\boldsymbol \lambda_N = N \ptrue$, $\boldsymbol t_N = \boldsymbol Y$, $\boldsymbol \kappa_N = N\ddot{\boldsymbol p}$. If $N \to \infty$, then $\|\boldsymbol \lambda_N\| \to \infty$ and both,
$$\boldsymbol t_N - \boldsymbol \lambda_N =  O_p(\|\boldsymbol \lambda_N\| ^{1/2}) \quad \mbox{ and } \quad \boldsymbol \kappa_N -  \boldsymbol \lambda_N =  O_p(\|\boldsymbol \lambda_N\| ^{1/2}),$$ 
which implies that 
\begin{equation} \label{ProbStochLemma}
\Delta[\boldsymbol \lambda_N^{-1}](\boldsymbol t_N - \boldsymbol \lambda_N) =  O_p(\|\boldsymbol \lambda_N\| ^{-1/2}) \, \mbox{ and } \, \Delta[\boldsymbol \lambda_N^{-1}](\boldsymbol \kappa_N - \boldsymbol \lambda_N) =  O_p(\|\boldsymbol \lambda_N\| ^{-1/2}).
\end{equation}
By Lemma   \ref{NonStochBregmanPearson2a}, applied to the non-stochastic version of (\ref {ProbStochLemma}) with $\boldsymbol y$ as a realizaton of $\boldsymbol Y$,
$$B(\boldsymbol y, N\ddot{\boldsymbol p}) -  X^2(\boldsymbol y, N\ddot{\boldsymbol p}) = o(1), \qquad \mbox{ as }  N \to \infty.$$
The latter also holds in the stochastic sense:
$$B(\boldsymbol Y, N\ddot{\boldsymbol p}) -  X^2(\boldsymbol Y, N\ddot{\boldsymbol p}) = o_p(1), \qquad \mbox{ as }  N \to \infty,$$
which completes the proof.
\end{proof}

\section{Empirical distributions of the test statistics}\label{SimulationsSection}

The results of a Monte Carlo simulation study presented here illustrate the finite sample behaviour of the test statistics considered in the paper. The models $RM(\mathbf{A_2})$ and $RM(\mathbf{A_3})$ were Aitchison-Silvey independence for intensities for 2 and 3 features, respectively. These are relational models generated by:

\begin{equation*}
\mathbf{A_2} = \left(\begin{tabular}{rrr}
1 & 0 & 1 \\ 
0 & 1 & 1 \end{tabular}\right), 
\qquad \mathbf{A_3} = \left(\begin{tabular}{rrrrrrr}
1 & 0 & 0 & 1 & 1 & 0 & 1  \\ 
0 & 1 & 0 & 1 & 0 & 1 & 1 \\ 
0 & 0 & 1 & 0 & 1 & 1 & 1  \\ 
 \end{tabular} \right). 
 \end{equation*}
The true distributions in each case, can be obtained from the subset effects, given in the first column of the corresponding table, according to the multiplicative structure specified by the model matrix. To illustrate, the true distribution under $RM(\mathbf{A_3})$ for the subset effects $(5,8,10)'$ is equal to $\boldsymbol \lambda =  (5,8,10,40,50,80,400)'$.
The empirical distributions of the statistics, shown in Table \ref{AS2sim}, 
were obtained based on 1000 simulations for selected choices of the parameter values. The MLEs under the models  were computed using the R package {\it gIPFrm}. For comparison, the charts also contain the chi-squared density with the appropriate number of degrees of freedom: 1 degree of freedom for the model $RM(\mathbf{A_2})$ and 4 degrees of freedom for $RM(\mathbf{A_3})$. The simulation results indicate that the Pearson and the Bregman statistics have empirical distributions close to the asymptotic even for small or moderate expectations of the sample sizes.

\begin{table}
\centering
\caption{The empirical distribution of the test statistics for Aitchison-Silvey Independence of two features, $RM(\mathbf{A_2})$ and of three features, $RM(\mathbf{A_3})$. }
\begin{tabular}{m{10mm}m{65mm}m{65mm}}
& &   \\
Effects &  \multicolumn{1}{c}{Bregman Statistic}  &  \multicolumn{1}{c}{Pearson Statistic}  \\ [2ex]
$\left(\begin{array}{c}5 \\10\end{array}\right)$&  \multicolumn{1}{c}{\begin{minipage}{.3\textwidth}\includegraphics[height=4cm,width=45mm]{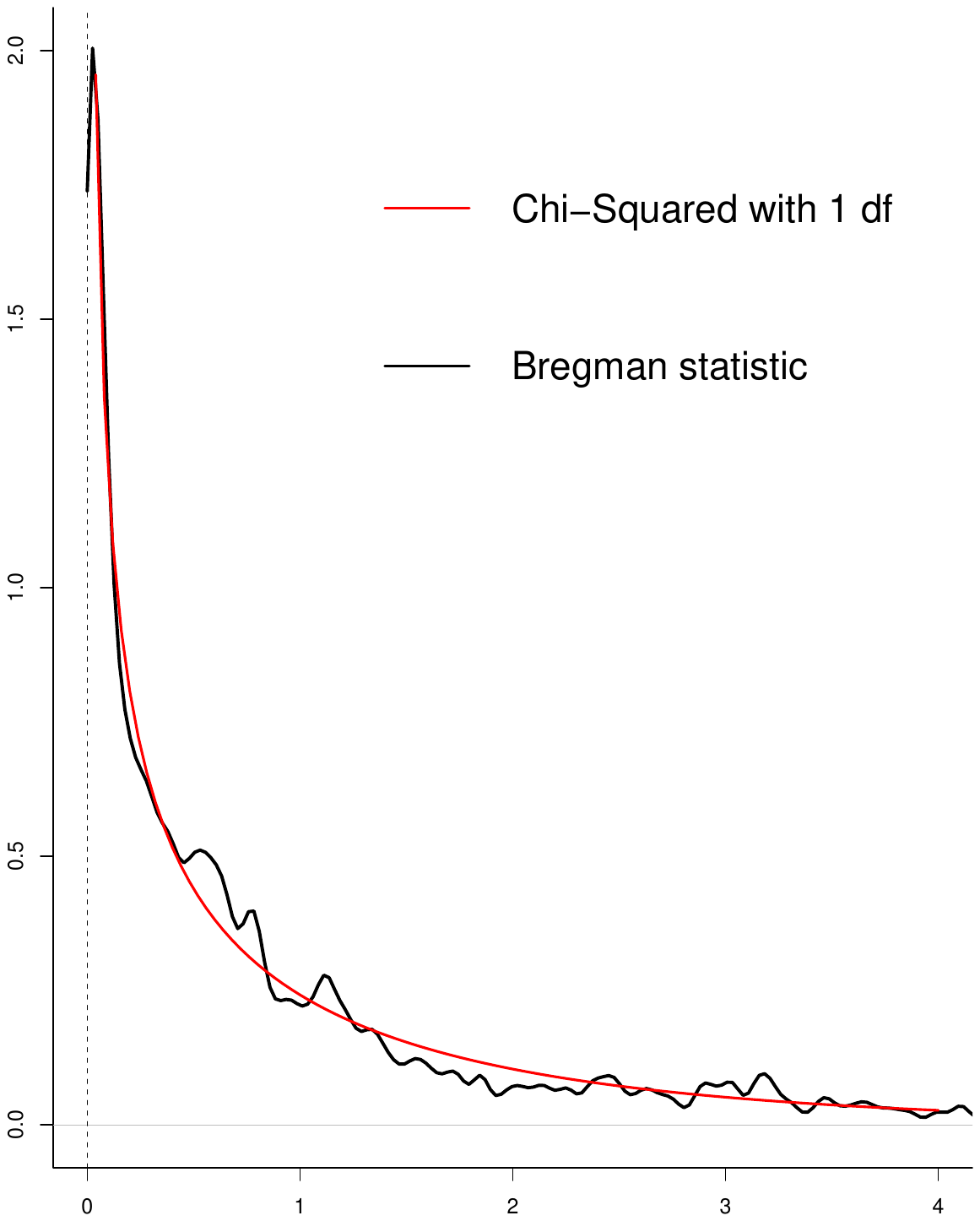}\end{minipage}}&    \multicolumn{1}{c}{\begin{minipage}{.3\textwidth}\includegraphics[height=4cm,width=45mm]{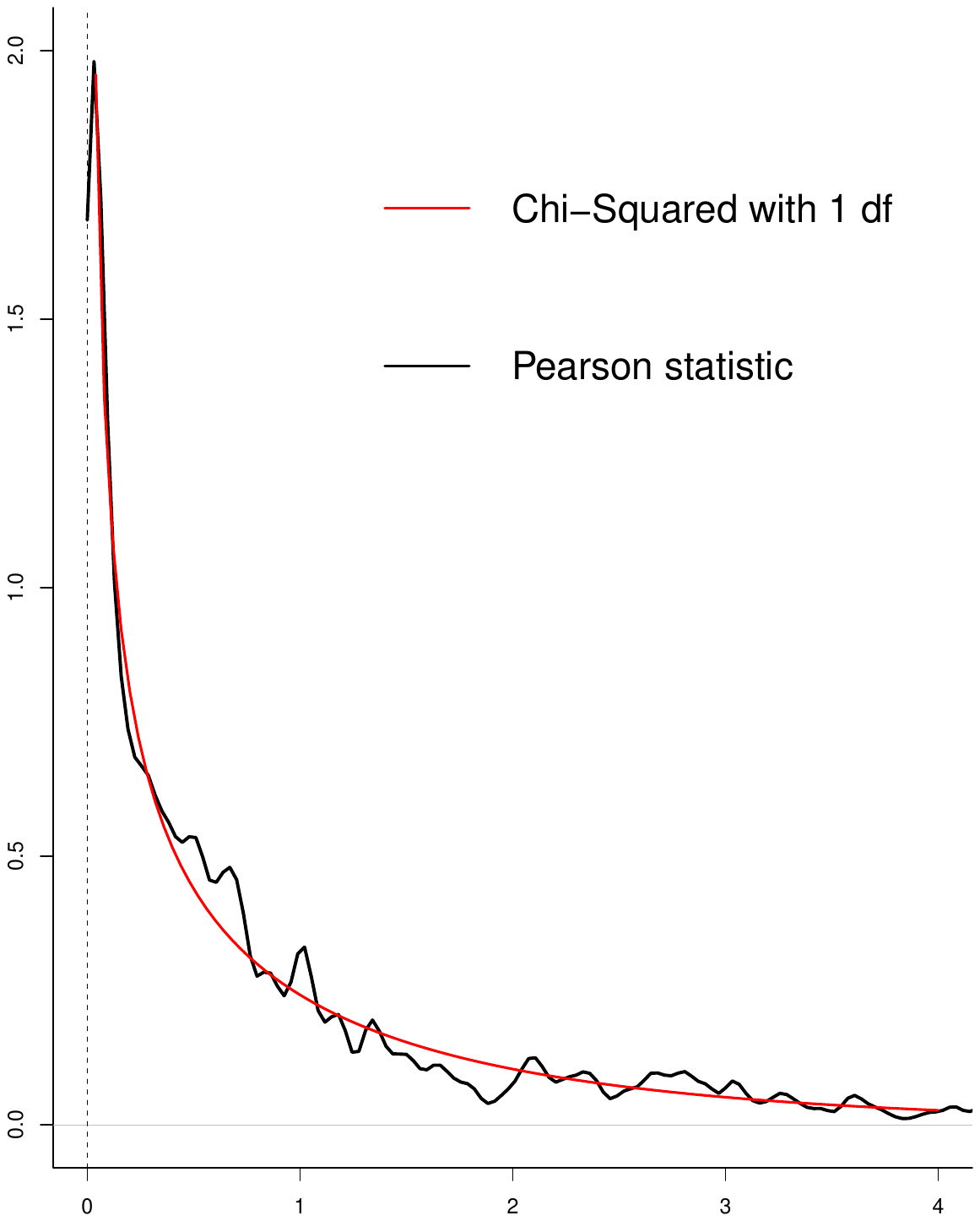}\end{minipage}}
 \\ [100pt]
 $\left(\begin{array}{c}20 \\20\end{array}\right)$&  \multicolumn{1}{c}{\begin{minipage}{.3\textwidth}\includegraphics[height=4cm,width=45mm]{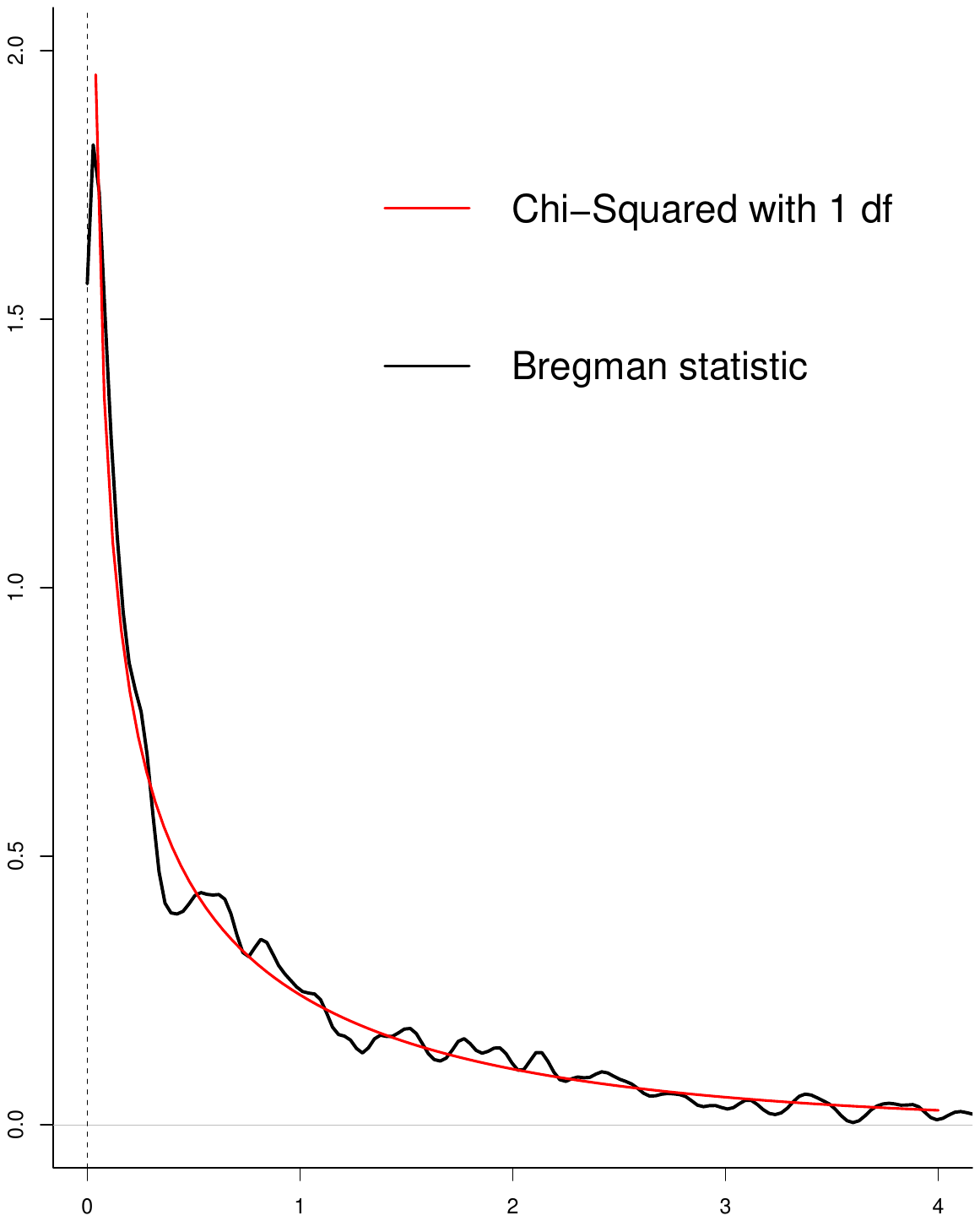}\end{minipage}}&  \multicolumn{1}{c}{\begin{minipage}{.3\textwidth}\includegraphics[height=4cm,width=45mm]{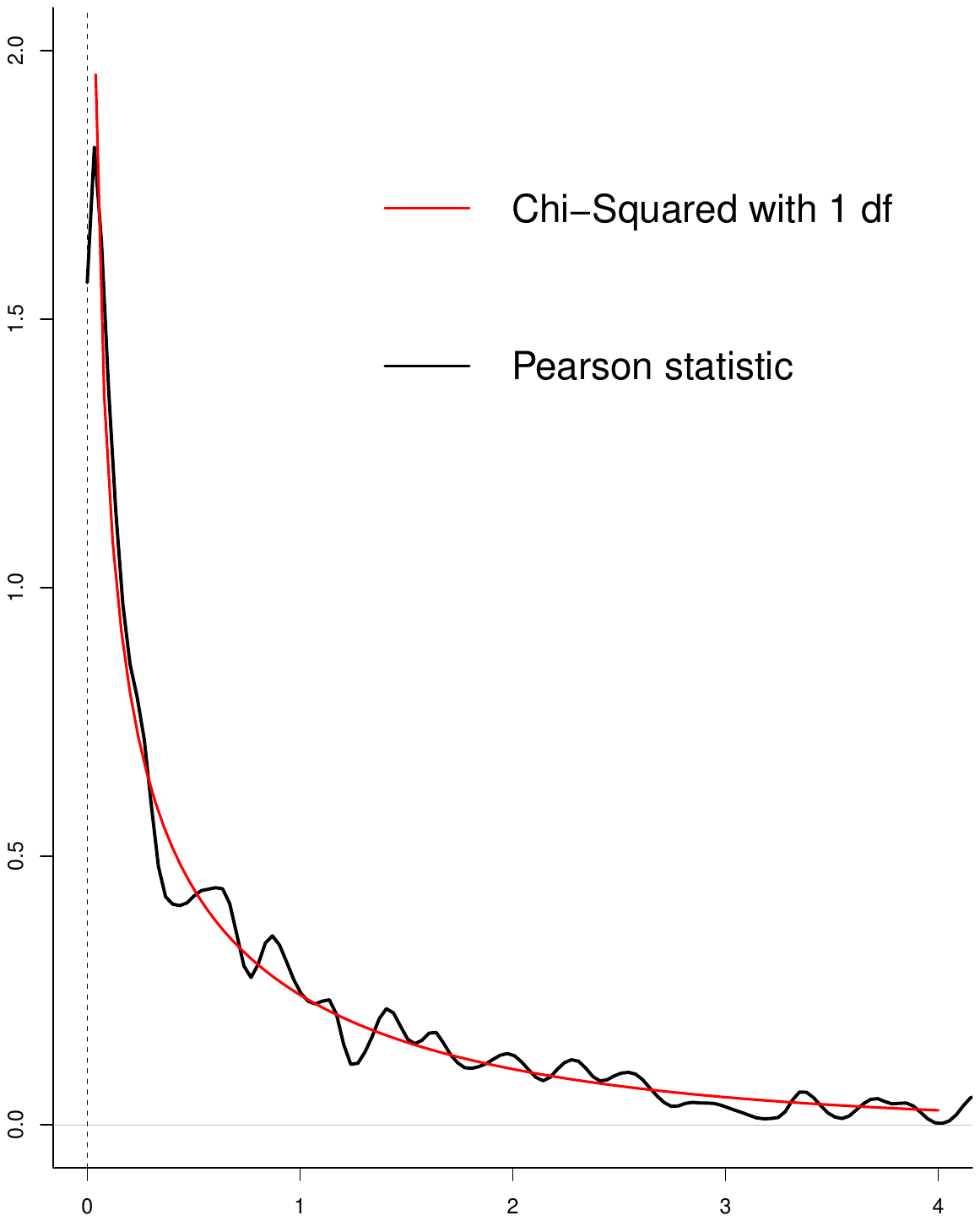}\end{minipage}}
 \\[100pt]
$\left(\begin{array}{c}5 \\8 \\10\end{array}\right)$ &  \multicolumn{1}{c}{\begin{minipage}{.3\textwidth}\includegraphics[height=45mm,width=45mm]{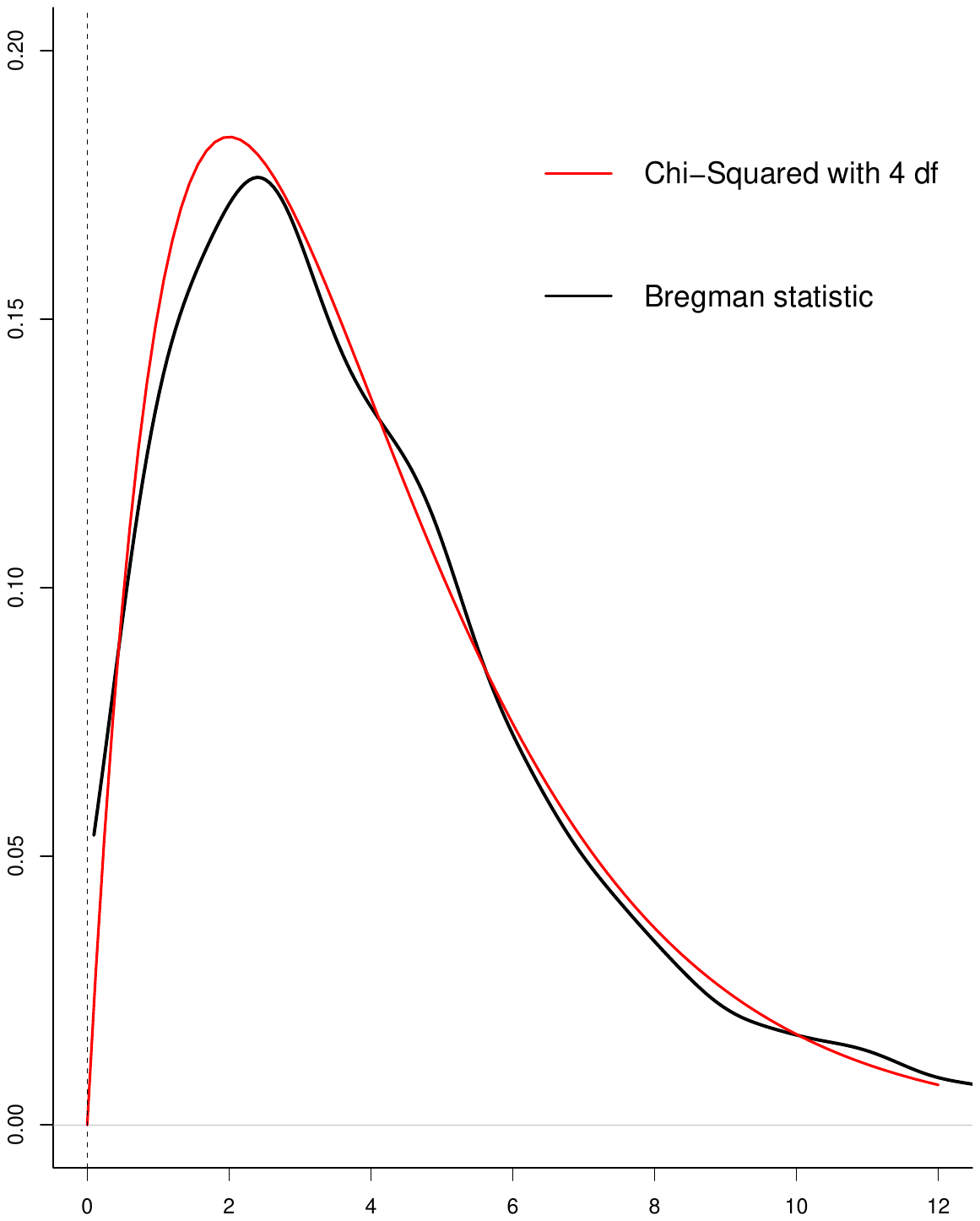}\end{minipage}}&    \multicolumn{1}{c}{\begin{minipage}{.3\textwidth}\includegraphics[height=45mm,width=45mm]{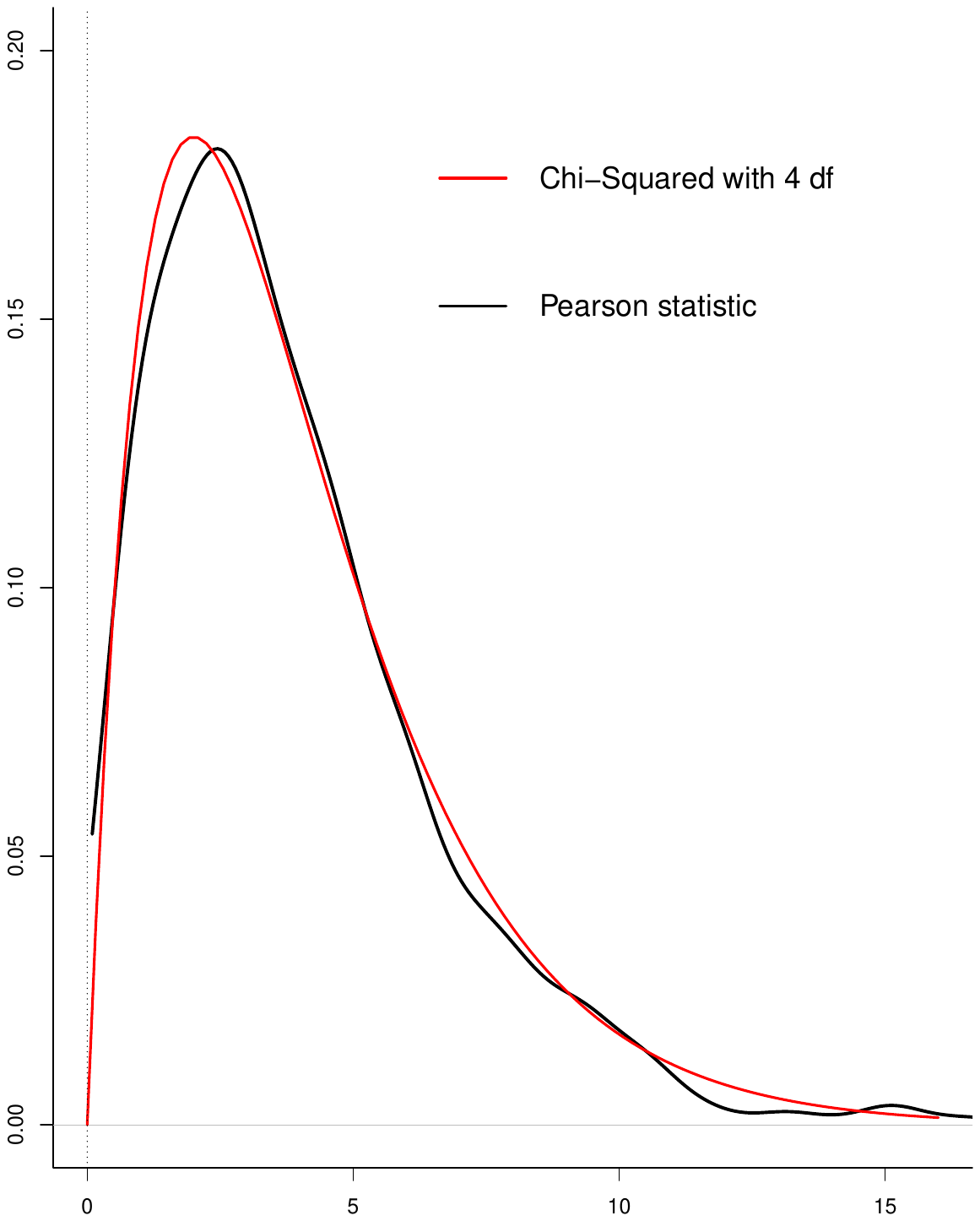}\end{minipage}}
  \\ [100pt]
  $\left(\begin{array}{c}10 \\10 \\20\end{array}\right)$ &  \multicolumn{1}{c}{\begin{minipage}{.3\textwidth}\includegraphics[height=45mm,width=45mm]{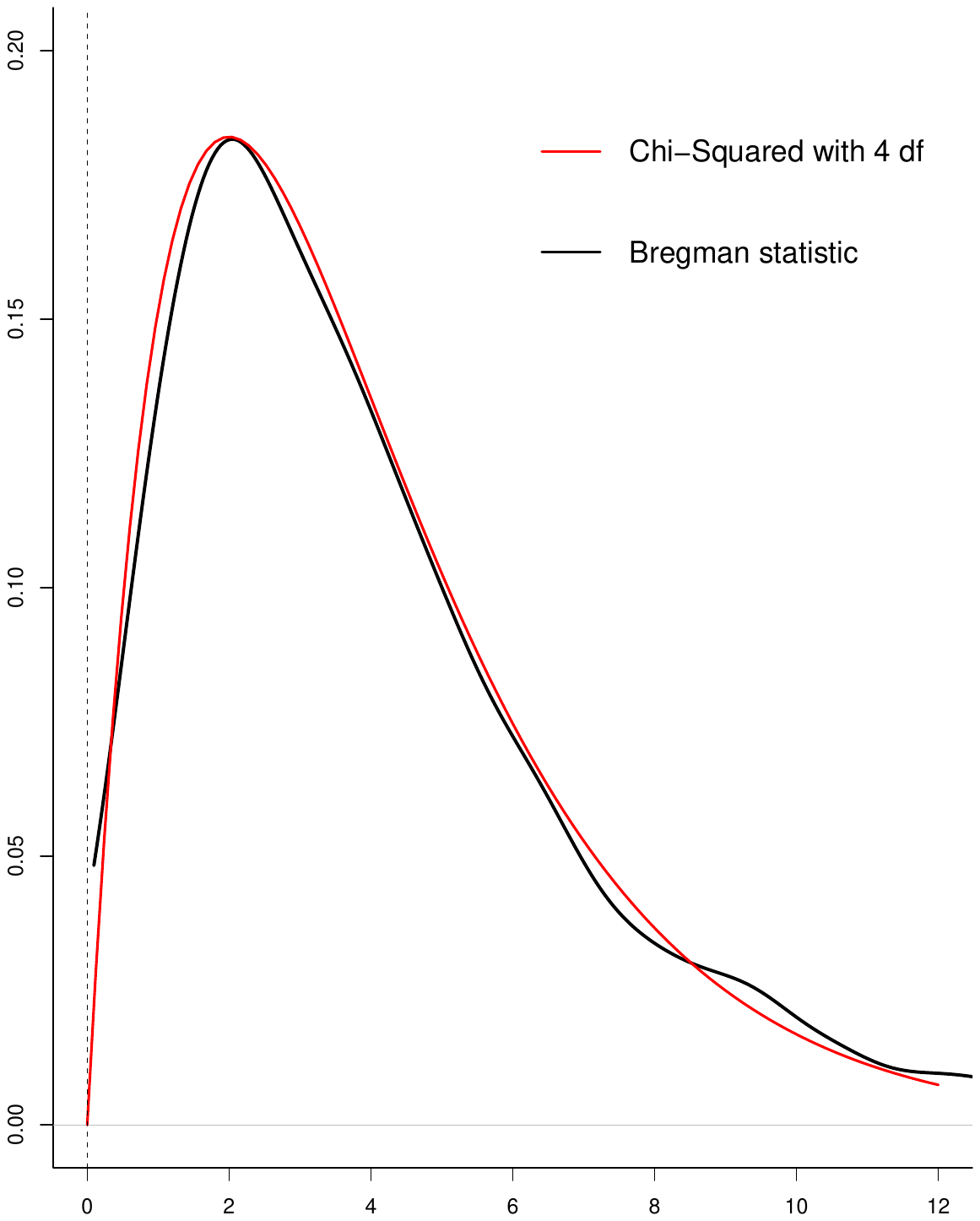}\end{minipage}}&  \multicolumn{1}{c}{\begin{minipage}{.3\textwidth}\includegraphics[height=45mm,width=45mm]{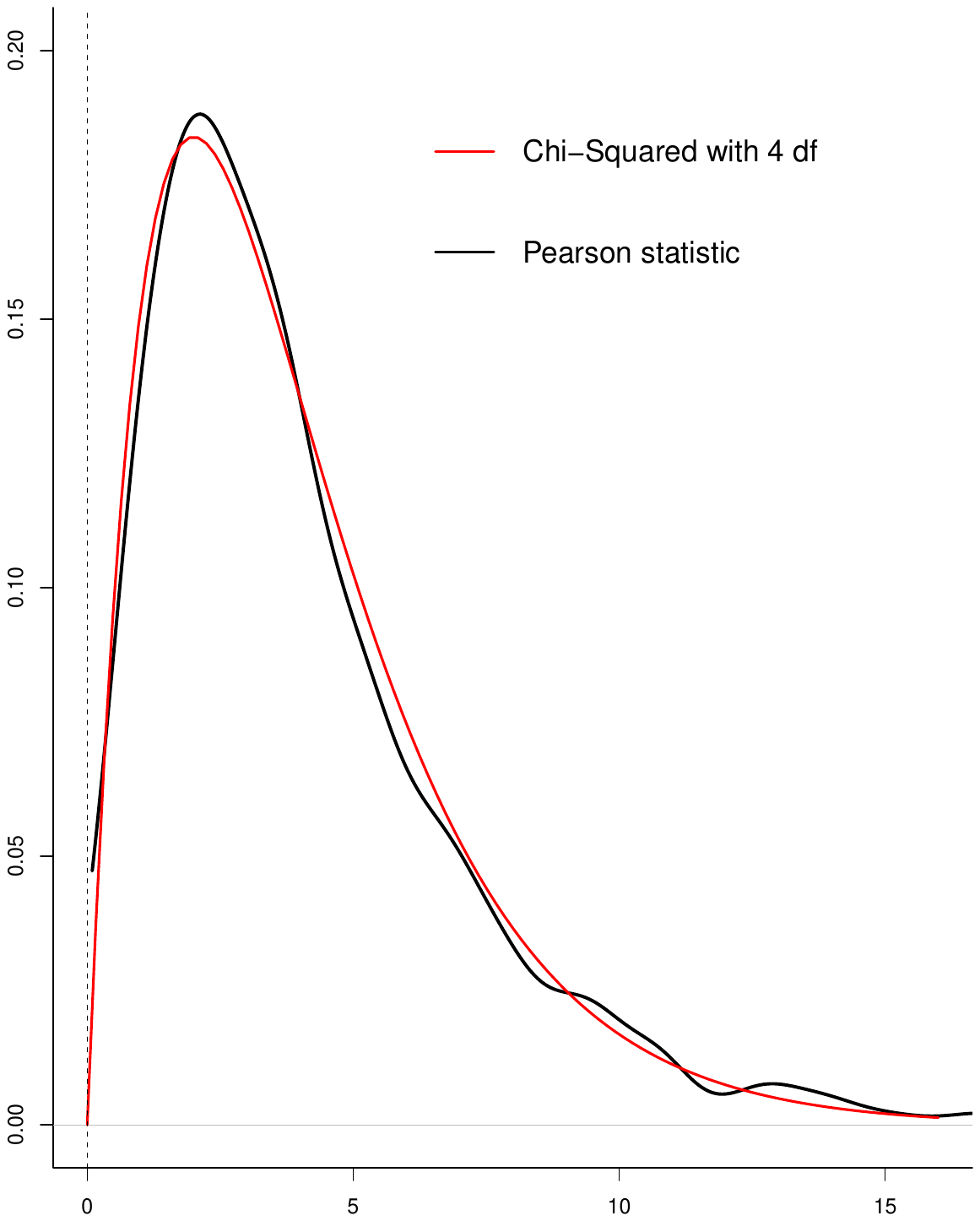}\end{minipage}}
\end{tabular}
\label{AS2sim}
\end{table}

\section*{Acknowledgments}

The second author was supported in part by Grant K-106154 from the Hungarian National Scientific Research Fund (OTKA). The authors wish to thank Wicher Bergsma for insightful comments.

\newpage
\appendix


\renewcommand{\thesection}{\Alph{section}}
\counterwithin{theorem}{section}
\renewcommand{\thetheorem}{\Alph{section}.\arabic{theorem}}

\section*{Appendix}


\vspace{4mm}

\noindent \textbf{Proof of Lemma \ref{newlemma}}

\vspace{3mm}

Let $\boldsymbol y$ be a realization of $\boldsymbol Y$. As shown in \cite{KRipf1}, the MLE $\hat{\boldsymbol \lambda}$ is the Bregman projection of $\boldsymbol y$ on the convex set defined as $
\{ \boldsymbol \lambda > \boldsymbol 0: \,\, \mathbf{A} \boldsymbol \lambda =  \mathbf{A} \boldsymbol y \}$. As follows from the inequality for the Bregman divergence given in Lemma 1 in \cite{Bregman}, 
$$
\mathcal{B}(\hat{\boldsymbol \lambda}, \boldsymbol y) \leq \mathcal{B}(\true, \boldsymbol y) -  \mathcal{B}({\true}, \hat{\boldsymbol \lambda}).$$
Notice that $\mathcal{B}(\hat{\boldsymbol \lambda}, \boldsymbol y) \geq 0$, and thus, the function $g(\true, \hat{\boldsymbol \lambda}, \boldsymbol y) = \mathcal{B}({\true}, \boldsymbol y) - \mathcal{B}({\true}, \hat{\boldsymbol \lambda})$ is non-negative.
Expand $g(\true, \boldsymbol y, \hat{\boldsymbol \lambda})$ in a Taylor series, assuming that $\|\true\| \to \infty$:
\begin{align*}
g(\true, \hat{\boldsymbol \lambda}, \boldsymbol y)
& = \sum_{i=1}^I \left\{\ti \log (\ti/y_i) - ( \ti - y_i) - \ti \log (\ti/\hat{\lambda}_i) + (\ti - 
\hat{\lambda}_i)\right\} \\
& = \sum_{i=1}^I \ti \left\{\log (\ti/y_i) - ( 1 - y_i/\ti) - \log (\ti/\hat{\lambda}_i) + (1 - \hat{\lambda}_i/\ti )\right\} \\
& = \sum_{i=1}^I \ti \left\{-\log (y_i/\ti) + ( y_i/\ti -1) + \log (\hat{\lambda}_i/\ti) - (\hat{\lambda}_i/\ti -1)\right\}\\
& = \sum_{i=1}^I \ti \left\{-\left((y_i/\ti - 1) - \frac{1}{2}(y_i/\ti - 1)^2 + o((y_i/\ti - 1)^2)\right) + ( y_i/\ti -1) \right. \\
&+ \left. \left((\hat{\lambda}_i/\ti - 1) - \frac{1}{2}(\hat{\lambda}_i/\ti - 1)^2 + o((\hat{\lambda}_i/\ti - 1)^2)\right) - 
(\hat{\lambda}_i/\ti -1 )\right\} \\
& = \sum_{i=1}^I \ti \left\{-\left(- \frac{1}{2}(y_i/\ti - 1)^2 + o((y_i/\ti - 1)^2)\right) \right. \\
&+ \left. \left( - \frac{1}{2}(\hat{\lambda}_i/\ti - 1)^2 + o((\hat{\lambda}_i/\ti - 1)^2)\right)\right\} \\
& = \sum_{i=1}^I \ti \left\{\frac{1}{2}(y_i/\ti - 1)^2 - o((y_i/\ti - 1)^2) - \frac{1}{2}(\hat{\lambda}_i/\ti - 1)^2 + o((\hat{\lambda}_i/\ti - 1)^2)\right\} .\\
\end{align*}
Because $g(\true, \hat{\boldsymbol \lambda}, \boldsymbol y) \geq 0$,

$$\sum_{i=1}^I \ti \left\{(\hat{\lambda}_i/\ti - 1)^2  +  o((\hat{\lambda}_i/\ti - 1)^2)\right\} \leq \sum_{i=1}^I \ti \left\{(y_i/\ti - 1)^2 + o((y_i/\ti - 1)^2)\right\},$$
or, equivalently,
$$\sum_{i=1}^I \ti{}^{-1} \left\{(\hat{\lambda}_i - \ti)^2  +  o((\hat{\lambda}_i - \ti)^2)\right\} \leq \sum_{i=1}^I \ti{}^{-1} \left\{(y_i - \ti )^2 + o((y_i - \ti )^2)\right\}.$$
Eq. (\ref{Op}) implies that $ \ti{}^{-1} (Y_i - \ti)^2 = O_p(1)$ for every $i =1, \dots, I$.  Thus, 
$$\sum_{i=1}^I \ti{}^{-1} \left\{(Y_i - \ti )^2 + o((Y_i - \ti)^2)\right\} = O_p(1),$$
and, because
\begin{align*}
\sum_{i=1}^I \ti{}^{-1/2} (\hat{\lambda}_i - \ti)  &\leq \left(\sum_{i=1}^I \ti{}^{-1} (\hat{\lambda}_i - \ti)^2\right)^{1/2},
\end{align*}
it also holds that $\sum_{i=1}^I \ti{}^{-1/2} (\hat{\lambda}_i - \ti) = O_p(1)$.

\hfill \qed
\vspace{4mm}

\noindent \textbf{Proof of Lemma \ref{smallLemma}}

\vspace{3mm}

Let $\boldsymbol \varsigma  > \boldsymbol 0$ be such that  $\|\boldsymbol \lambda - \boldsymbol  \varsigma\| < \|\boldsymbol \lambda - \true\|$, and therefore,
$\boldsymbol  \varsigma = t \cdot\boldsymbol \lambda + (1-t)\cdot\true$  for some $t\in (0,1)$. Then,
\begin{align*}
\Delta[\true{}^{-1}]\boldsymbol \varsigma &= t\cdot \Delta[\true{}^{-1}]\boldsymbol \lambda + (1-t)\cdot \Delta[\true{}^{-1}]\true  =t\cdot \Delta[\true{}^{-1}](\boldsymbol \lambda - \true) + \boldsymbol 1, \end{align*}
and
\begin{align*}
\Delta[\true{}^{-1}]\boldsymbol \varsigma - \boldsymbol 1 &=t\cdot \Delta[\true{}^{-1}](\boldsymbol \lambda - \true) = \Delta[\true{}^{-1/2}]\cdot O_p(\boldsymbol 1) = O_p(\|\true\|^{-1/2}), 
\end{align*}
as $\|\true\| \to \infty$. Component wise, 
$$\frac{\lambda_{0i}}{\varsigma_{i}} - 1 = O_p(\lambda_{0i}^{-1/2}), \qquad \mbox{ for } i \in \mathcal{I}.$$
Using the latter, as $\lambda_{0i} \to \infty$,
\begin{align*}
\lambda_{0i}/ \varsigma_i &= \frac{1}{(\varsigma_i/\lambda_{0i} - 1)+1} = 1 - ( \varsigma_i/\lambda_{0i} - 1) + o_p(| \varsigma_i/\lambda_{0i} - 1|);\\
&{}\\
(\lambda_{0i}/ \varsigma_i)^m &= (1 - ( \varsigma_i/\lambda_{0i} - 1) + o_p(| \varsigma_i/\lambda_{0i} - 1|))^m\\
&{}\\
&= 1 - m(\varsigma_i/\lambda_{0i} - 1) + o_p(|\varsigma_i/\lambda_{0i} - 1|) = 1 + O_p(\lambda_{0i}^{-1/2}).
\end{align*}
Therefore, 
\begin{align*}
\Delta[\true{}^m] \cdot \Delta[\boldsymbol \varsigma^{-m}] &= \mathbf{I}+ O_p(\|\true\|^{-1/2}), \qquad \mbox{as } \|\true\| \to \infty.
\end{align*}
\hfill\qed

\vspace{4mm}

\noindent \textbf{Proof of Corollary \ref{CovarianceOE}}

\vspace{3mm}
Notice that 
\begin{align} \label{Hinverse}
\mathbf{H}'
\Delta[\ptrue]\mathbf{H} &=  \left(\boldsymbol 1, \Delta[\ptrue{}^{-1}]\mathbf{D}'\right)'\Delta[\ptrue]\left(\boldsymbol 1, \Delta[\ptrue{}^{-1}]\mathbf{D}'\right) 
=  \left(\begin{array}{c} \boldsymbol 1' \\ \mathbf{D} \Delta[\ptrue{}^{-1}] \end{array} \right)\Delta[\ptrue]\left(\boldsymbol 1, \Delta[\ptrue{}^{-1}]\mathbf{D}'\right) \nonumber\\[3pt]
&=  \left(\begin{array}{c} \boldsymbol 1' \\ \mathbf{D} \Delta[\ptrue{}^{-1}] \end{array} \right)\left(\ptrue,\mathbf{D}'\right) 
=  \left(\begin{array}{cc} 1 & \boldsymbol 1'\mathbf{D}' \\ \mathbf{D} \boldsymbol 1 &\mathbf{D} \Delta[\ptrue{}^{-1}] \mathbf{D}'\end{array} \right).
\end{align}
\vspace{2mm}
If the overall effect is present, then $\mathbf{D} \boldsymbol 1 = \boldsymbol 0$, and (\ref{Hinverse}) can be simplified. 
In this case,
$$(\mathbf{H}'
\Delta[\ptrue]\mathbf{H})^{-1} = \left(\begin{array}{cc} 1 & \boldsymbol 0' \\ \boldsymbol 0 & (\mathbf{D} \Delta[\ptrue{}^{-1}] \mathbf{D}')^{-1}\end{array} \right), \quad \mbox{and}$$
\begin{align*}
\mathbf{M} &= \mathbf{I} - \Delta[\ptrue] \cdot \mathbf{H}(\mathbf{H}'
\Delta[\ptrue]\mathbf{H})^{-1}\mathbf{H}'\\
&= \mathbf{I} - \left(\ptrue,\mathbf{D}'\right)\cdot \left(\begin{array}{cc} 1 & \boldsymbol 0' \\ \boldsymbol 0 & (\mathbf{D} \Delta[\ptrue{}^{-1}] \mathbf{D}')^{-1}\end{array} \right) \cdot \left(\begin{array}{c} \boldsymbol 1' \\ \mathbf{D} \Delta[\ptrue{}^{-1}] \end{array} \right) \\
&{}\\
&= \mathbf{I} - \left(\ptrue,\mathbf{D}'\right)\cdot \left(\begin{array}{c} \boldsymbol 1' \\ (\mathbf{D} \Delta[\ptrue{}^{-1}] \mathbf{D}')^{-1} \mathbf{D} \Delta[\ptrue{}^{-1}] \end{array} \right) 
= \mathbf{I} - \ptrue \cdot \boldsymbol 1' - \mathbf{D}'(\mathbf{D} \Delta[\ptrue{}^{-1}] \mathbf{D}')^{-1}  \mathbf{D} \Delta[\ptrue{}^{-1}] \\[10pt]
&=  \Delta[\ptrue] \Delta[\ptrue{}^{-1}]  - \ptrue \cdot \ptrue' \Delta[\ptrue{}^{-1}] - \mathbf{D}'(\mathbf{D} \Delta[\ptrue{}^{-1}] \mathbf{D}')^{-1}  \mathbf{D} \Delta[\ptrue{}^{-1}] \\[10pt]
&= \left(\boldsymbol \Sigma - \mathbf{D}'(\mathbf{D} \Delta[\ptrue{}^{-1}] \mathbf{D}')^{-1}  \mathbf{D}\right)\cdot \Delta[\ptrue{}^{-1}].
\end{align*}
The covariance matrix of the asymptotic distribution can be written as
\begin{align*}
\mathbf{M} \boldsymbol \Sigma \mathbf{M}'= &\left(\boldsymbol \Sigma - \mathbf{D}'(\mathbf{D} \Delta[\ptrue{}^{-1}] \mathbf{D}')^{-1}  \mathbf{D}\right)\cdot (\Delta[\ptrue{}^{-1}] - \boldsymbol 1 \cdot \boldsymbol 1') \cdot\left(\boldsymbol \Sigma - \mathbf{D}'(\mathbf{D} \Delta[\ptrue{}^{-1}] \mathbf{D}')^{-1}  \mathbf{D}\right),
\end{align*}
which can be simplified further, using that:
\begin{align*}
&\mathbf{D}'(\mathbf{D} \Delta[\ptrue{}^{-1}] \mathbf{D}')^{-1}  \mathbf{D}\boldsymbol 1 = \boldsymbol 0, \hspace{10mm}
\boldsymbol \Sigma\cdot  \Delta[\ptrue{}^{-1}]  \cdot \boldsymbol \Sigma = \boldsymbol \Sigma, \hspace{10mm}
\boldsymbol \Sigma \boldsymbol 1 = \boldsymbol 0.
\end{align*}
Finally,
$$\mathbf{M} \boldsymbol \Sigma \mathbf{M}'= \boldsymbol \Sigma - \mathbf{D}'(\mathbf{D} \Delta[\ptrue{}^{-1}] \mathbf{D}')^{-1}  \mathbf{D}. $$
\hfill \qed

\vspace{5mm}

\noindent \textbf{Proof of Lemma \ref{NonStochBregmanPearsona} }

\vspace{4mm}


Let $\boldsymbol t_r =(t_{r1}, \dots, t_{rI})$ and $\boldsymbol \lambda_{r} = (\lambda_{r1}, \dots, \lambda_{rI})$. It is sufficient to show that for each $i = 1, \dots, I$,
$$B(t_{ri}, \lambda_{ri}) - X^2(t_{ri}, \lambda_{ri}) = o(1), \qquad \mbox{as } \lambda_{ri} \to \infty.$$
For simplicity of exposition, $t_r \equiv t_{ri}$ and $\lambda_r \equiv \lambda_{ri}$.

Expand $B(t_r,\lambda_r)$ in a Taylor series in terms of ${t_r}/{\lambda_r} - 1$:
\begin{align*}
B(t_r,\lambda_r) &= 2\cdot(t_r\log\left(\frac{t_r}{\lambda_r}\right) - (t_r - \lambda_r)) \\
&=2\cdot\lambda_r \left\{\frac{t_r}{\lambda_r} \log \left(1 + \left(\frac{t_r}{\lambda_r}-1\right)\right) - \left(\frac{t_r}{\lambda_r}-1\right)\right\}\\
&= 2\cdot \lambda_r \left\{\frac{t_r}{\lambda_r}\left(\frac{t_r}{\lambda_r}-1\right) - \frac{t_r}{2\lambda_r}\left(\frac{t_r}{\lambda_r}-1\right)^2 + o(|\frac{t_r}{\lambda_r}-1|^2) - \left(\frac{t_r}{\lambda_r}-1\right)\right\}\\
&= 2\cdot \lambda_r \left\{\left(\frac{t_r}{\lambda_r}-1\right)^2 - \frac{t_r}{2\lambda_r}\left(\frac{t_r}{\lambda_r}-1\right)^2 + o(|\frac{t_r}{\lambda_r}-1|^2) \right\} \\
&= 2\cdot \lambda_r \left\{\left(\frac{t_r}{\lambda_r}-1\right)^2 -\frac{1}{2}\left(\frac{t_r}{\lambda_r}-1\right)^2 - \frac{1}{2}\left (\frac{t_r}{\lambda_r}-1\right)\left(\frac{t_r}{\lambda_r}-1\right)^2 + o(|\frac{t_r}{\lambda_r}-1|^2) \right\} 
\end{align*}

\begin{align*}
&= \lambda_r \left(\frac{t_r}{\lambda_r}-1\right)^2 + \lambda_r \cdot o(|\frac{t_r}{\lambda_r}-1|^2) = \frac{(t_r - \lambda_r)^2}{\lambda_r} + o(\lambda_r|\frac{t_r}{\lambda_r}-1|^2) \\
&= \frac{(t_r - \lambda_r)^2}{\lambda_r} + o(\lambda_rO(\lambda_r^{-1})) = \frac{(t_r - \lambda_r)^2}{\lambda_r} + o(1). 
\end{align*}
\hfill\qed
\vspace{5mm}

\noindent \textbf{Proof of Lemma \ref{NonStochBregmanPearson2a}}
\vspace{4mm}

Let $\boldsymbol t_n =(t_{n1}, \dots, t_{nI})$, $\boldsymbol \kappa_n =(\kappa_{n1}, \dots, \kappa_{nI})$, and $\boldsymbol \lambda_{n} = (\lambda_{n1}, \dots, \lambda_{nI})$. It is sufficient to show that for each $i = 1, \dots, I$,
$$B(t_{ni}, \kappa_{ni}) - X^2(t_{ni}, \kappa_{ni}) = o(1), \qquad \mbox{as } \lambda_{ni} \to \infty.$$
For simplicity of presentation, $t \equiv t_{ni}$, $\kappa \equiv \kappa_{ni}$, and $\lambda \equiv \lambda_{ni}$
\begin{align*}
&{} B(t, \kappa) - X^2(t, \kappa) = 2[t \log \frac{t}{\kappa} - (t - \kappa)] - \frac{(t - \kappa)^2}{\kappa}\\
&{} \\
&= 2[t \log \frac{t\cdot \lambda}{\kappa \cdot\lambda} - (t - \lambda - (\kappa- \lambda))] - \frac{(t - \lambda - (\kappa - \lambda))^2}{\kappa} \\
&= 2\left\{\lambda\frac{t}{\lambda} \log \frac{t}{\lambda} - \lambda(\frac{t}{\lambda} - 1) \right\} - 2\left\{t\log \frac{\kappa}{\lambda} - \lambda(\frac{\kappa}{\lambda} - 1) \right\} 
- \frac{\lambda^2}{\kappa}\left\{\frac{t}{\lambda} - 1 - (\frac{\kappa}{\lambda} - 1) \right\}^2 \\
&{} \\
&= \lambda (\frac{t}{\lambda} -1)^2 + o(1) - 2\left\{t\log( 1 + (\frac{\kappa}{\lambda} - 1)) - \lambda(\frac{\kappa}{\lambda} - 1)\right\} \\
&- \frac{\lambda^2}{\kappa}\left\{(\frac{t}{\lambda} - 1)^2 - 2 (\frac{t}{\lambda} - 1)(\frac{\kappa}{\lambda} - 1) + (\frac{\kappa}{\lambda} - 1)^2 \right\}
&{} \\
&= \lambda (\frac{t}{\lambda} -1)^2 + o(1) - 2\left\{t\log( 1 + (\frac{\kappa}{\lambda} - 1)) - \lambda(\frac{\kappa}{\lambda} - 1)\right\} \\
&- \frac{\lambda^2}{\kappa}\left\{(\frac{t}{\lambda} - 1)^2 - 2 (\frac{t}{\lambda} - 1)(\frac{\kappa}{\lambda} - 1) + (\frac{\kappa}{\lambda} - 1)^2 \right\}
&{} \\
&= \lambda (\frac{t}{\lambda} -1)^2 + o(1) - 2\left\{t\left[ (\frac{\kappa}{\lambda} - 1) - \frac{1}{2}(\frac{\kappa}{\lambda} - 1)^2 + o((\frac{\kappa}{\lambda} - 1)^2)\right] - \lambda(\frac{\kappa}{\lambda} - 1)\right\} \\
&- \frac{\lambda^2}{\kappa}\left\{(\frac{t}{\lambda} - 1)^2 - 2 (\frac{t}{\lambda} - 1)(\frac{\kappa}{\lambda} - 1) + (\frac{\kappa}{\lambda} - 1)^2 \right\}\\
&{} \\
&= \lambda (\frac{t}{\lambda} -1)^2 + o(1) - \left\{\left[ 2t(\frac{\kappa}{\lambda} - 1) - t(\frac{\kappa}{\lambda} - 1)^2 + 2t\cdot o((\frac{\kappa}{\lambda} - 1)^2)\right] - 2\lambda(\frac{\kappa}{\lambda} - 1)\right\} \\
&- \frac{\lambda^2}{\kappa}\left\{(\frac{t}{\lambda} - 1)^2 - 2 (\frac{t}{\lambda} - 1)(\frac{\kappa}{\lambda} - 1) + (\frac{\kappa}{\lambda} - 1)^2 \right\}
\end{align*}

\begin{align*}
&= \lambda (\frac{t}{\lambda} -1)^2 + o(1) - 2t(\frac{\kappa}{\lambda} - 1) + t(\frac{\kappa}{\lambda} - 1)^2 - 2t\cdot o((\frac{\kappa}{\lambda} - 1)^2)+ 2\lambda(\frac{\kappa}{\lambda} - 1) \\
&- \frac{\lambda^2}{\kappa}\left\{(\frac{t}{\lambda} - 1)^2 - 2 (\frac{t}{\lambda} - 1)(\frac{\kappa}{\lambda} - 1) + (\frac{\kappa}{\lambda} - 1)^2 \right\}\\
&{} \\
&= (\frac{t}{\lambda} -1)^2 \left\{ \lambda - \frac{\lambda^2}{\kappa}\right\} + (\frac{\kappa}{\lambda} - 1) \left\{ -2t + 2\lambda \right\} + (\frac{\kappa}{\lambda} - 1)^2\left\{t - \frac{\lambda^2}{\kappa}\right\} \\
& + 2\frac{\lambda^2}{\kappa} (\frac{t}{\lambda} - 1)(\frac{\kappa}{\lambda} - 1)+ o(1) - 2t\cdot o((\frac{\kappa}{\lambda} - 1)^2) \\
&{} \\
&= \lambda (\frac{t}{\lambda} -1)^2 \left\{1- \frac{\lambda}{\kappa}\right\} + 2\lambda (\frac{\kappa}{\lambda} - 1) \left\{ -\frac{t}{\lambda} + 1 \right\} + \lambda(\frac{\kappa}{\lambda} - 1)^2\left\{\frac{t}{\lambda} - \frac{\lambda}{\kappa}\right\} \\
&{} \\
& + 2\frac{\lambda^2}{\kappa} (\frac{t}{\lambda} - 1)(\frac{\kappa}{\lambda} - 1) + o(1) - 2t\cdot o((\frac{\kappa}{\lambda} - 1)^2)\\
&{} \\
&= \lambda (\frac{t}{\lambda} -1)^2 \left\{1- \frac{\lambda}{\kappa}\right\} + (-2\lambda +2\frac{\lambda^2}{\kappa}) (\frac{\kappa}{\lambda} - 1) \left\{ \frac{t}{\lambda} - 1 \right\} + \lambda(\frac{\kappa}{\lambda} - 1)^2\left\{\frac{t}{\lambda} - \frac{\lambda}{\kappa}\right\} \\
& + o(1) - 2t\cdot o((\frac{\kappa}{\lambda} - 1)^2) \\
&{} \\
&= -\lambda (\frac{t}{\lambda} -1)^2 \left\{-1+ \frac{\lambda}{\kappa}\right\} +2\lambda (-1+ \frac{\lambda}{\kappa}) (\frac{\kappa}{\lambda} - 1) \left\{ \frac{t}{\lambda} - 1 \right\} + \lambda(\frac{\kappa}{\lambda} - 1)^2\left\{\frac{t}{\lambda} - \frac{\lambda}{\kappa}\right\} \\
 &+ o(1) - 2t\cdot o((\frac{\kappa}{\lambda} - 1)^2)\\
 &{} \\
 &= -\lambda (\frac{t}{\lambda} -1)^2 \left\{\frac{\lambda}{\kappa}-1\right\} +2\lambda (\frac{\lambda}{\kappa}-1) (\frac{\kappa}{\lambda} - 1) \left\{ \frac{t}{\lambda} - 1 \right\} \\
&+ \lambda(\frac{\kappa}{\lambda} - 1)^2\left\{\frac{t}{\lambda} - 1 -(\frac{\lambda}{\kappa}-1)\right\} 
+ o(1) - 2t\cdot o((\frac{\kappa}{\lambda} - 1)^2)\\
&{}\\
&= -\lambda (\frac{t}{\lambda} -1)^2 (\frac{\lambda}{\kappa}-1) +2\lambda (\frac{\lambda}{\kappa}-1) (\frac{\kappa}{\lambda} - 1) (\frac{t}{\lambda} - 1 ) \\
&+ \lambda(\frac{\kappa}{\lambda} - 1)^2(\frac{t}{\lambda} - 1) - \lambda(\frac{\kappa}{\lambda} - 1)^2(\frac{\lambda}{\kappa}-1)+ o(1) - 2t\cdot o((\frac{\kappa}{\lambda} - 1)^2)\\
&{}\\
&= -\lambda (\frac{\lambda}{\kappa}-1)\left\{ (\frac{t}{\lambda} -1)^2 -2 (\frac{\kappa}{\lambda} - 1) (\frac{t}{\lambda} - 1) + (\frac{\kappa}{\lambda} - 1)^2\right\} \\
&+ \lambda(\frac{\kappa}{\lambda} - 1)^2(\frac{t}{\lambda} -1) 
 + o(1) - 2t\cdot o((\frac{\kappa}{\lambda} - 1)^2)
\end{align*}

\begin{align*}
 &= \lambda (1-\frac{\lambda}{\kappa})\left\{ (\frac{t}{\lambda} -1) - (\frac{\kappa}{\lambda} - 1)\right\}^2 - \lambda(\frac{\kappa}{\lambda} - 1)^2 + o(1) + \lambda\frac{t}{\lambda}(\frac{\kappa}{\lambda} - 1)^2 - 2\lambda\frac{t}{\lambda}\cdot o((\frac{\kappa}{\lambda} - 1)^2)\\
&{} \\
&= \lambda (1-\frac{\lambda}{\kappa})\left\{ (\frac{t}{\lambda} -1) - (\frac{\kappa}{\lambda} - 1)\right\}^2 - \lambda(\frac{\kappa}{\lambda} - 1)^2 \\
& + o(1) + \lambda\cdot(1+O(\lambda^{-1/2})) \cdot(\frac{\kappa}{\lambda} - 1)^2 - 2\lambda\cdot(1+O(\lambda^{-1/2})) \cdot o((\frac{\kappa}{\lambda} - 1)^2)
\end{align*}

\vspace{2mm}

To evaluate the above expression, the following facts will be used.
For the components of the above expression:
$$\kappa/\lambda -1 = O(\lambda^{-1/2}); \,\, t/\lambda -1 = O(\lambda^{-1/2}),$$
$$1-\frac{\lambda}{\kappa} = 1-\frac{1}{\kappa/\lambda} = 1-\frac{1}{1 + ({\kappa}/{\lambda}  -1)} = 1 - [1- ({\kappa}/{\lambda} - 1)] + o(|\kappa/\lambda -1|) =(\kappa/\lambda -1) + o(|\kappa/\lambda -1|), $$
and therefore,
$$1-\frac{\lambda}{\kappa} = O(\lambda^{-1/2}) + o(O(\lambda^{-1/2})) = O(\lambda^{-1/2}).$$
Substituting these in the last expression, one obtains:
\begin{align*}
&\lambda \cdot O(\lambda^{-1/2}) \cdot[O(\lambda^{-1/2})]^2 - \lambda \cdot[O(\lambda^{-1/2})]^2 + o(1) \\
&+ \lambda\cdot(1+O(\lambda^{-1/2})) \cdot[O(\lambda^{-1/2}]^2 - 2\lambda\cdot(1+O(\lambda^{-1/2})) \cdot o([O(\lambda^{-1/2}]^2) \\
&= O(\lambda^{-1/2}) + o(1) + o(1) + o(1)\cdot (1+O(\lambda^{-1/2})) + \lambda \cdot o(\lambda^{-1}) = o(1).
\end{align*}
\hfill \qed

\bibliographystyle{plainnat}

\bibliography{referencesStatSim}

\end{document}